\documentclass[11pt,english,preprintnumbers,amsmath,amssymb,superscriptaddress,nofootinbib,prl]{revtex4}
\usepackage{mathpazo}
\usepackage[T1]{fontenc}
\usepackage[latin9]{inputenc}
\usepackage[letterpaper]{geometry}
\usepackage{verbatim}
\usepackage{amsmath}
\usepackage{amsthm}
\usepackage{amssymb}
\usepackage{graphicx}

\makeatletter
\@ifundefined{textcolor}{}
{%
 \definecolor{BLACK}{gray}{0}
 \definecolor{WHITE}{gray}{1}
 \definecolor{RED}{rgb}{1,0,0}
 \definecolor{GREEN}{rgb}{0,1,0}
 \definecolor{BLUE}{rgb}{0,0,1}
 \definecolor{CYAN}{cmyk}{1,0,0,0}
 \definecolor{MAGENTA}{cmyk}{0,1,0,0}
 \definecolor{YELLOW}{cmyk}{0,0,1,0}
}
  \theoremstyle{plain}
  \newtheorem*{thm*}{\protect\theoremname}
  \theoremstyle{remark}
  \ifx\thechapter\undefined
    \newtheorem{rem}{\protect\remarkname}
  \else
    \newtheorem{rem}{\protect\remarkname}[chapter]
  \fi
  \theoremstyle{definition}
  \ifx\thechapter\undefined
    \newtheorem{defn}{\protect\definitionname}
  \else
    \newtheorem{defn}{\protect\definitionname}[chapter]
  \fi
  \theoremstyle{plain}
  \ifx\thechapter\undefined
    \newtheorem{lem}{\protect\lemmaname}
  \else
    \newtheorem{lem}{\protect\lemmaname}[chapter]
  \fi

\makeatother

\usepackage{babel}
  \providecommand{\definitionname}{Definition}
  \providecommand{\lemmaname}{Lemma}
  \providecommand{\remarkname}{Remark}
  \providecommand{\theoremname}{Theorem}

\begin{document}

\title{The gap of Fredkin quantum spin chain is polynomially small}

\author{Ramis Movassagh }
\email{q.eigenman@gmail.com}

\selectlanguage{english}%

\affiliation{Department of Mathematical Sciences, IBM T. J . Watson Research Center,
Yorktown Heights, NY 10598}

\date{\today}
\begin{abstract}
We prove a new result on the spectral gap and mixing time of a Markov
chain with Glauber dynamics on the space of Dyck paths (i.e., Catalan
paths) and their generalization, which we call colored Dyck paths.

Let $2n$ be the number of spins. We prove that the gap of the Fredkin
quantum spin chain Hamiltonian \cite{salberger2016fredkin,dell2016violation}, 
is $\Theta(n^{-c})$ with $c\ge2$. Our results on the spectral gap
of the Markov chain are used to prove a lower bound of $O(n^{-15/2})$ 
on the energy of first excited state above the ground state of the Fredkin quantum spin chain. We prove an upper
bound of $O(n^{-2})$ using the universality of Brownian motion and
convergence of Dyck random walks to Brownian excursions. Lastly, the
'unbalanced' ground state energies are proved to be polynomially small
in $n$ by mapping the Hamiltonian to an effective hopping Hamiltonian
with next nearest neighbor interactions and analytically solving
its ground state.
\end{abstract}
\maketitle
\tableofcontents{}
\section{\label{sec:Context-and-summary}Context and Summary of the Results}
In recent years there has been a surge of activities in developing
new exactly solvable models that violate the area law for the entanglement
entropy \cite{Movassagh2012_brackets,movassagh2016supercritical,dell2016violation,salberger2016fredkin}.
The notion of exactly solvable in these works means that the ground
state can be written down analytically and the gap to the first excited
state is quantified. Understanding the gap is important for the physics
of quantum many-body systems and in particular these models.

In \cite{movassagh2016supercritical} a surprising new class of exactly solvable
quantum spin chain models were proposed. These models have positive
integer spins ($s>1$), the Hamiltonian is nearest neighbors with
a unique ground state that can be seen as a uniform superposition
of $s-$colored Motzkin walks defined below. The half-chain entanglement entropy
provably violates the area law by a square root factor in the system's
size ($\sim\sqrt{n}$). The power-law violation of the entanglement
entropy in that work provides a counter-example to the widely believed
notion, that translationally invariant spin chains with a unique ground
state and local interactions can violate the area law by at most a
logarithmic factor in the system's size. 

In \cite{Movassagh2012_brackets,movassagh2016supercritical} the gap to the
first excited state was quantified and it was found that in both cases
the gap vanishes as a polynomial in the system's size. And in \cite{movassagh2017entanglement}
the Hamiltonian for $s=1$ was expressed in standard spin basis, and
certain physically relevant quantities such as the spin correlation
functions in the ground state, and the block von Neumann and Renyi
entanglement entropies of the ground state were analytically calculated.

Recently, Olof Salberger and Vladimir Korepin extended the model presented in \cite{movassagh2016supercritical}  to half-integer spin chains and named the new model Fredkin spin chain \cite{salberger2016fredkin,dell2016violation}.
The proposed Hamiltonian is $3-$local, the ground state is unique
and may be seen as a uniform superposition of $s-$colored Dyck walks.
For $s=1/2$ and $s>1/2$, the half-chain entanglement entropy violates
the area law by $\sim\log(n)$ and $\sim\sqrt{n}$, respectively. Although
certainly gapless, the recent papers \cite{dell2016violation,salberger2016fredkin}
did not obtain the gap of their model. We prove that the gap of
Fredkin spin chain \cite{dell2016violation,salberger2016fredkin}
vanishes as $n^{-c}$ with $c\ge2$.  

The area law has been rigorously proved in one-dimensional gapped systems \cite{Matth_areal} . Moreover, when the ground state is unique but the energy gap vanishes in the thermodynamical limit, it is expected that the area-law conjecture still holds, but now with a possible logarithmic correction, which in one-dimension implies an entanglement entropy of $O(\log n)$. In other words, as long as the ground state is unique, the area-law can be violated by at most a logarithmic factor. This belief is mainly based on analysis of $1+1$-dimensional critical system that in the continuum are describable by conformal field theories as well as Fermi liquids. In relativistic conformal field theories the elementary excitations (e.g., gap) scale as $1/n$. A corollary of the $O(1/n^2)$ upper bound on the energy gap in this work, therefore, is that the Fredkin spin-chain also does not have  a relativistic conformal field theory continuum limit. 

Section \ref{sec:New-Results-MC} defines the so called Fredkin \textit{Markov}
chain whose state space is the set of all Dyck paths (Catalan paths)
and their generalization which we call colored Dyck paths. In section
\ref{sec:New-Results-MC} we prove a lower-bound on the spectral gap
of Fredkin Markov chain and hence an upper-bound on its mixing time.
The main mathematical techniques used are the comparison theorem of
Diaconis and Saloff-Coste \cite{diaconis1993comparison} and our previous
results \cite{Movassagh2012_brackets,movassagh2016supercritical}. Section
\ref{sec:New-Results-MC} is independent, purely mathematical and
divorced from any ``quantum'' notions. Therefore it can be read
independently.

In Section \ref{sec:Hamiltonian-Gap-of} we discuss the general mapping
between certain quantum Hamiltonians and classical Markov chains and
give an upper-bound on the gap of Fredkin quantum spin chain that
is polynomially small in the system's size. We do so by ``twisting''
the states in the ground state using phases that are proportional
to the area under the Dyck paths, and then utilize universality of
Brownian motion, and convergence of Dyck random walks to Brownian
excursions to give the desired upper-bound. A polynomially small lower-bound
on the gap is obtained in two steps: 1. Using the results of Section
\ref{sec:New-Results-MC}, the gap of Fredkin spin chain is lower-bounded
in a particular subspace, which we denote as the ``balanced'' subspace.
2. The smallest eigenvalue (ground state) of the Hamiltonian restricted
to the complement of that subspace (i.e., ``Unbalanced'' subspace)
is lower bounded to be polynomially small by mapping the Hamiltonian
onto an effective hopping Hamiltonian with next nearest neighbor interactions and explicitly solving its ground
state and ground state energy.

\section{\label{sec:New-Results-MC}New Results on Markov Chain Mixing Times}
This section is self-contained and includes a new mathematical result
on the mixing times of certain random walk Markov chains. The presentation
in this section, therefore, solely follows the mathematical literature
of Markov chains and can be read independently of the other sections.

\textbf{1. Preliminaries on Markov Chains.} Let $P(x,y)$ be an irreducible
Markov chain on the state space $\Omega$ with the stationary probability
distribution $\pi$. Everywhere we take the stationary distribution
to be uniform, all the graphs undirected, and assume the Markov chain
is reversible
\[
\pi(x)P(x,y)=\pi(y)P(y,x).
\]
$P$ has eigenvalues $1=\lambda_{0}>\lambda_{1}\ge\cdots\ge\lambda_{|\Omega|-1}\ge-1$.
The gap of the Markov chain is defined to be $1-\lambda_{1}$. 

When a Markov chain starts in state $x_{0}$ and runs for $t$ steps,
we denote the state at time $t$ by $P_{x_{0}}^{t}$. If the Markov
chain is connected and aperiodic, then as $t\rightarrow\infty$ the
distribution of $P_{x_{0}}^{t}$ converges to the unique stationary
distribution $\pi$. The measure of the distance between $P_{x_{0}}^{t}$
and $\pi$ is usually given by the total variational distance defined
by
\[
\Delta_{x_{0}}(t)\equiv||P_{x_{0}}^{t}-\pi||_{TV}=\max_{A\in\Omega}|P_{x_{0}}^{t}(A)-\pi(A)|=\frac{1}{2}\sum_{y\in\Omega}|P_{x_{0}}^{t}(y)-\pi(y)|=\frac{1}{2}\left\Vert P_{x_{0}}^{t}-\pi\right\Vert _{1}.
\]

Following \cite{randall2000analyzing} we define the relationship
between the mixing time and the second eigenvalue of the Markov chain
matrix. For $\epsilon>0$, the \textit{mixing time}, starting from
state $x_{0}$ is defined by 
\[
\tau_{x_{0}}(\epsilon)=\min\{t:\text{ }\Delta_{x_{0}}(t')\le\epsilon,\quad\forall t'\ge t\}.
\]

Everywhere we refer to \textit{mixing time }as the time starting from
the worst case, i.e., 
\[
\tau(\epsilon)=\max_{x\in\Omega}\tau_{x}(\epsilon)\qquad\text{mixing time}
\]

The mixing time is related to $\lambda_{1}$ by 
\begin{eqnarray}
\tau_{x}(\epsilon) & \le & \frac{1}{1-\lambda_{1}}\log\left(\frac{1}{\pi(x)\epsilon}\right)\quad\forall x\in\Omega\nonumber \\
\tau(\epsilon) & \ge & \frac{\lambda_{1}}{2(1-\lambda_{1})}\log(1/2\epsilon)\quad.\label{eq:MixingTime}
\end{eqnarray}

\textbf{2. Preliminaries on random walks. }The paths considered in
this paper are taken to be paths in the standard cartesian $xy-$plane.
A \textit{Dyck path} from coordinates $(0,0)$ to $(2n,0$) is a path
with steps $(1,1)$ and $(1,-1)$ that never passes below the $x-$axis.
In other words a Dyck path on $2n$ steps is a walk from left to right
between the coordinates $\left(x_{0},y_{0}\right)=\left(0,0\right)$
and $\left(x_{2n},y_{2n}\right)=\left(2n,0\right)$. And at each step
the $y-$coordinate changes by one $(x_{i+1},y_{i+1})=(x_{i}+1,y_{i}\pm1)$,
and  $y_{i}\ge0$ for all $i$. 
\begin{figure}
\begin{centering}
\includegraphics[scale=0.35]{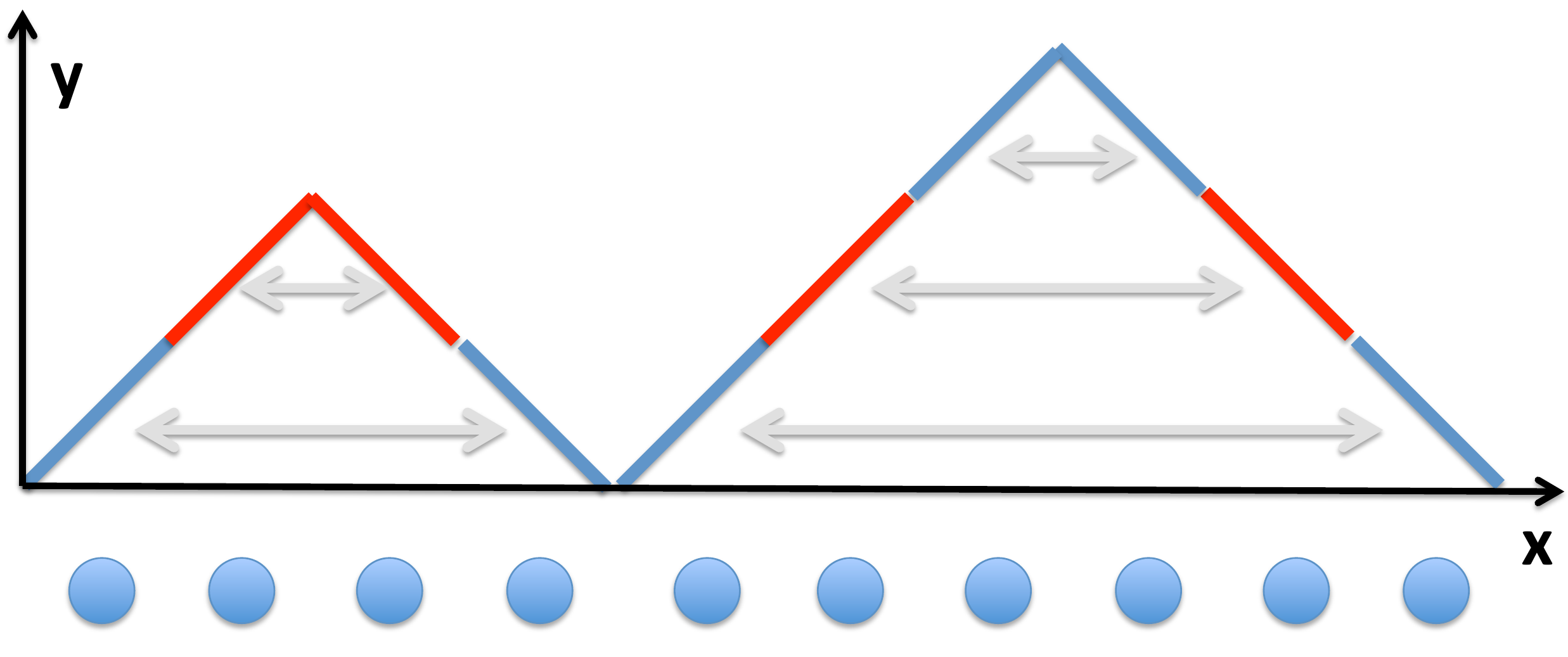}
\par\end{centering}
\centering{}\caption{\label{fig:Dyck}A Dyck walk with $s=2$ colors on a chain of length
$2n=10$.}
\end{figure}
Let $s$ be a positive integer. An $s-$\textit{colored Dyck path}
from $(0,0)$ to $(2n,0$) is a Dyck path whose  steps $(1,1)$ and
$(1,-1)$ each can have $s\ge1$ different colors. In addition, every
step of a given color is uniquely matched with a step down of the
same color. Note that $s=1$ corresponds to the standard Dyck paths
(see Fig. \eqref{fig:Dyck} for an example). 

Below, at times, we denote a step up by $u\equiv(1,1)$ and a step down
by $d\equiv(1,-1)$. When there are colorings, we use $u^{k}$ and
$d^{k}$ to denote a step up and down of color $k$, where $1\le k\le s$. 

A \textit{peak} on the walk is defined to be any coordinate whose
height (i.e., $y-$coordinate) is larger than both of its neighbors.

\textbf{3. Fredkin Markov Chain.} The Markov chain in this paper and
it predecessors \cite{Movassagh2012_brackets,movassagh2016supercritical}
have local transition rules. Such Markov chains are said to have
Glauber dynamics. Let us describe the uncolored Fredkin Markov chain
first as it is more mainstream in the combinatorics community. The
state space is the set of Dyck paths, and the Markov chain has the
local transition rules shown in Fig. \eqref{fig:localMoves-1}. These
moves interchange a peak with a step up or down to its left or right.
It is important to emphasize that any given Dyck walk can be related
to any other by a series of these local moves. Moreover, applying
any sequence of these local moves to any given Dyck results in a sequence
of Dyck walks only. 
\begin{figure}
\begin{centering}
\includegraphics[scale=0.4]{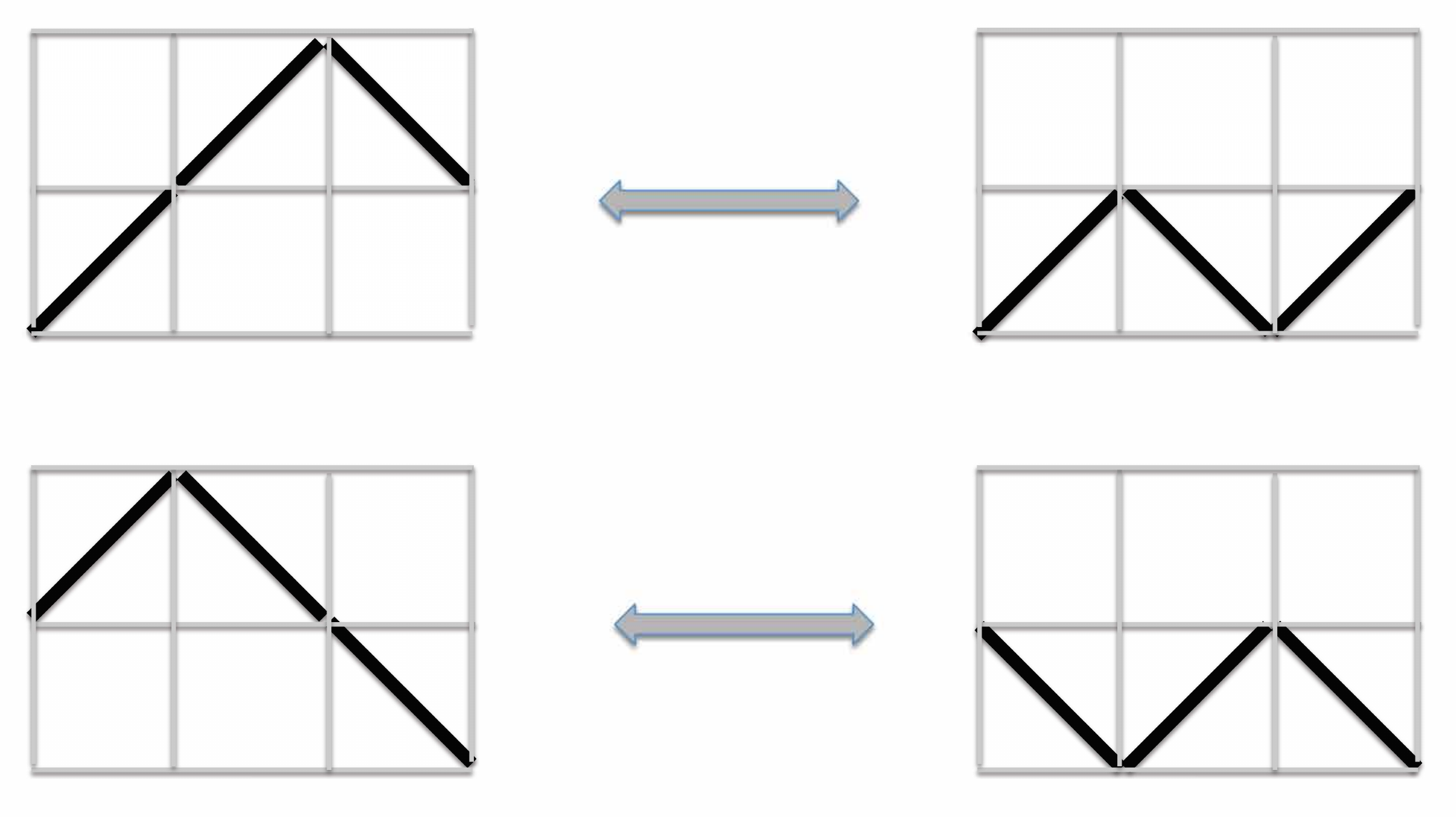}
\par\end{centering}
\centering{}\caption{\label{fig:localMoves-1}Local transition rules for the Fredkin Markov
chain.}
\end{figure}

The $s-$colored version is entirely similar except that any peak
of a \textit{given} color can be interchanged with any step up or
down step of whatever color. Moreover, there is a transition rule
that changes the color of a given peak. Some of these transition rules
are  shown in Fig. \eqref{fig:localMoves_Colored}. 
\begin{figure}
\begin{centering}
\includegraphics[scale=0.4]{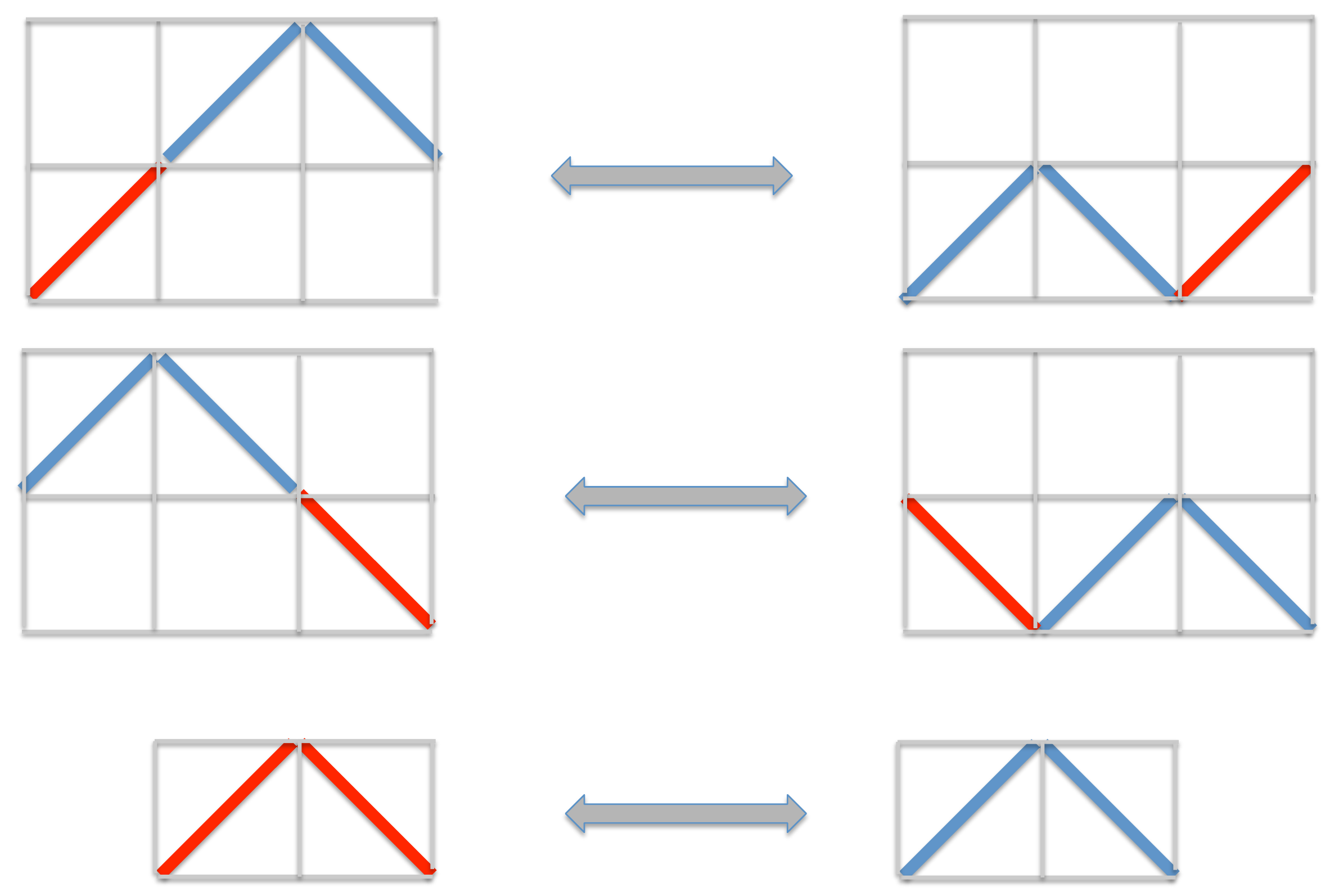}
\par\end{centering}
\centering{}\caption{\label{fig:localMoves_Colored}Local transition rules for the $2-$colored
Fredkin Markov chain. We are only showing exchange of a blue peak
with red steps to convey the transition rules with an understanding
that there are other moves. Namely, a peak of any given color (here
blue or red) can be exchanged with a step down or up of any color
(here blue or red). }
\end{figure}

\textbf{4. Lower-bound on the gap of Fredkin Markov Chain.} We wish
to lower bound the spectral gap of the \textit{Fredkin Markov chain}.
Previously we considered another Markov chain on the space of $s-$colored
Dyck paths, where the transition rules were as follows:
\begin{enumerate}
\item Pick a position between $1$ and $2n-1$ on the Dyck path at random. 
\item If there is a peak there, remove it to get a path of length $2n-2$. 
\item Insert a peak at a random position between $0$ and $2n-2$ with a color
randomly chosen uniformly out of the $s$ possibilities. 
\end{enumerate}
The spectral gap of this Markov chain was proved to be $1-\lambda_{1}=\frac{s}{\sqrt{\pi}n^{11/2}}$
\cite{movassagh2016supercritical}. Let us call this Markov chain \textit{peak-displacing.}

We use the comparison theorem of Diaconis and Saloff-Coste \cite{diaconis1993comparison}
to lower bound the gap of the $s-$colored Fredkin Markov chain in
terms of the known gap of the peak-displacing Markov chain. Let us first
describe the comparison theorem. We follow the presentation of \cite{diaconis1993comparison}.

Suppose $P,\pi$ and $\tilde{P},\tilde{\pi}$ are two reversible Markov
chains that are supported on the same state space $\Omega$. We think
of $P,\pi$ as the chain whose mixing time or gap is unknown and desired,
and think of $\tilde{P},\tilde{\pi}$ as the chain with known eigenvalues.
For each pair of states $x\ne y$ with $\tilde{P}(x,y)>0$, fix a
sequence of steps $x=x_{0},x_{1},x_{2},\dots,x_{k}=y$ with $P(x_{i},x_{i+1})>0$.
This sequence of steps defines a path $\gamma_{x,y}$ on the state
space with the length $|\gamma_{x,y}|=k$. Let the set of edges in
$P$ be $E=\{(x,y);\text{ }P(x,y)>0\}$ and the set of edges in $\tilde{P}$
be $\tilde{E}=\{(x,y);\text{ }\tilde{P}(x,y)>0\}$. Moreover, let
$\tilde{E}(e)=\{(x,y)\in\tilde{E};\text{ }e\in\gamma_{x,y}\}$, where
$e\in E$. In words, $\tilde{E}(e)$ is the set of paths that contain
the edge $e$. 

In our case, $\tilde{P}$ is the peak-displacing Markov chain  and
$P$ is the Fredkin Markov chain. 
\begin{thm*}
(Comparison theorem \cite{diaconis1993comparison}, \cite{randall2000analyzing})
Let $\tilde{P},\tilde{\pi}$ and $P,\pi$ be reversible Markov chains
on a finite set $\Omega$. Then
\[
\left(1-\lambda_{1}(P)\right)\ge\frac{1}{A}\left(1-\lambda_{1}(\tilde{P})\right),
\]
where 
\begin{equation}
A=\max_{z\ne w:P(z,w)\ne0}\left\{ \frac{1}{\pi(z)P(z,w)}\sum_{x\ne y:(z,w)\in\gamma_{x,y}}|\gamma_{x,y}|\tilde{\pi}(x)\tilde{P}(x,y)\right\} .\label{eq:A}
\end{equation}
\end{thm*}
\begin{rem}
The quantity $A$ depends on the choice of canonical paths $\{\gamma_{x,y}\}$
that we are free to choose. They play an important role similar to
the canonical path technique that was introduced by Jerrum and Sinclair
\cite{Sinclair1992}. However, the notion of paths in the comparison
theorem is fundamentally different because the paths relate two different
Markov chains. 
\end{rem}
Both the peak-displacing and Fredkin Markov chains have the $s-$colored
Dyck paths as their vertex sets. They both have the uniform distribution
of all ($s-$colored) Dyck walks as their stationary distribution.
However, they have different edge sets $E$ and $\tilde{E}$. 
\begin{lem}
The spectral gap of Fredkin Markov chain is upper bounded by 
\begin{equation}
1-\lambda_{1}(P)\ge{O(s n^{-15/2})\text{ }\text{min}_{z\ne w}P(z,w)}.\label{eq:Gap_Markov_Final}
\end{equation}
\end{lem}
\begin{proof}
Since in our case $\pi(x)=\tilde{\pi}(x)=s^{-n}/C_{n}$, where $C_{n}$
is the $n^{\text{th}}$ Catalan number we have
\begin{figure}
\begin{centering}
\includegraphics[scale=0.35]{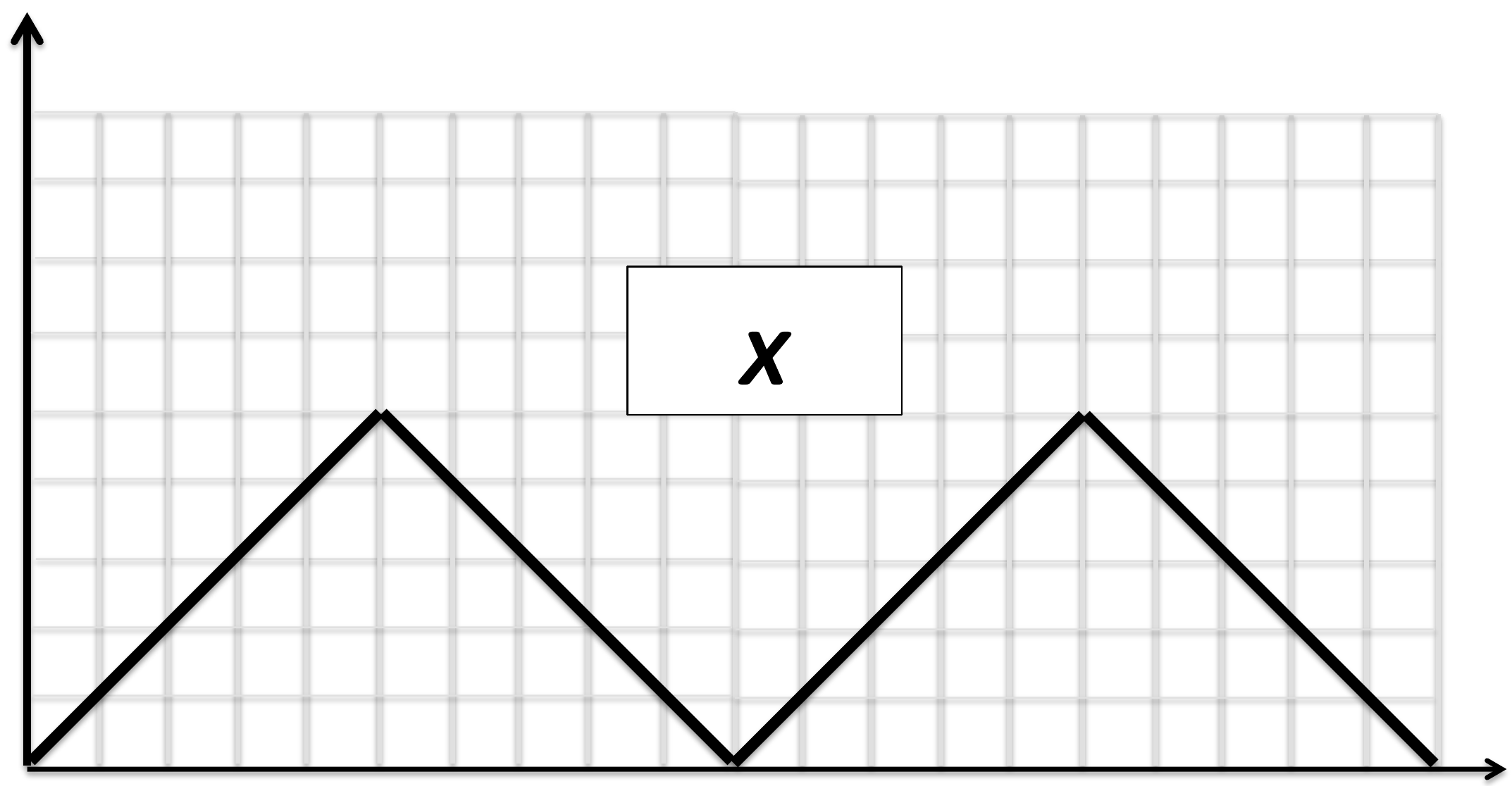}\includegraphics[scale=0.35]{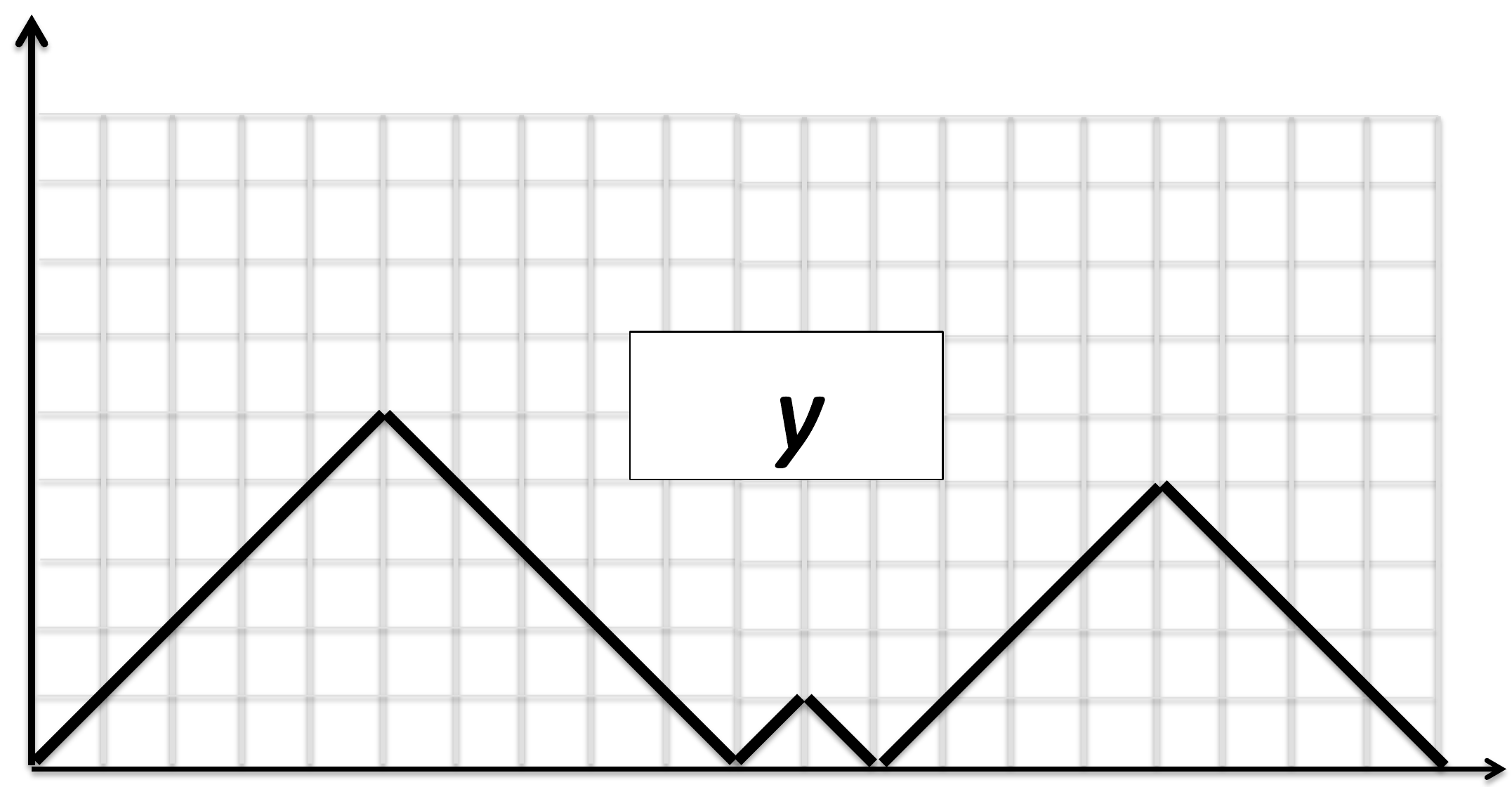}
\par\end{centering}
\caption{\label{fig:Two-adjacent-states}Two adjacent states in peak-displacing
Markov chain.}
\end{figure}
\[
1-\lambda_{1}(P)\ge A^{-1}\left\{ \frac{s}{\sqrt{\pi}n^{11/2}}\right\} .
\]

We now define the set of canonical paths. First take the uncolored
version $s=1$. Let $x$ and $y$ be two Dyck paths of length $2n$,
where $y$ is obtained from $x$ by randomly cutting a peak and inserting
it at a random position. An example of this is shown in Fig. \eqref{fig:Two-adjacent-states}, where the right most peak in $x$ was cut and placed in the valley.
Let the initial position of the peak (in $x$) be between $i$ and
$i+1$ and let its final position be between $j$ and $j+1$ (in $y$).
With no loss of generality take $j>i$ and define $k=j-i$. We define
the canonical path to be $x=x_{0},x_{1},\dots x_{k}=y$, where each
$x_{\ell}$ is a Dyck path that is obtained from $x_{\ell-1}$ by
exchanging the peak with a step to its right. This results in ``walking''
the peak $k$ steps to the right until $x_{k}=y$. Fig. \eqref{fig:The-canonical-path}
shows the canonical path that takes $x$ to $y$, that are adjacent in $P$, shown in Fig. \eqref{fig:Two-adjacent-states}.

Turning to the colored case, suppose the peak to be displaced has
different colors in $x$ and $y$. This can easily be implemented
in the canonical path by adding a single step at the end of the path
which changes the color of the peak at its final position.
\begin{figure}
\begin{centering}
\includegraphics[scale=0.35]{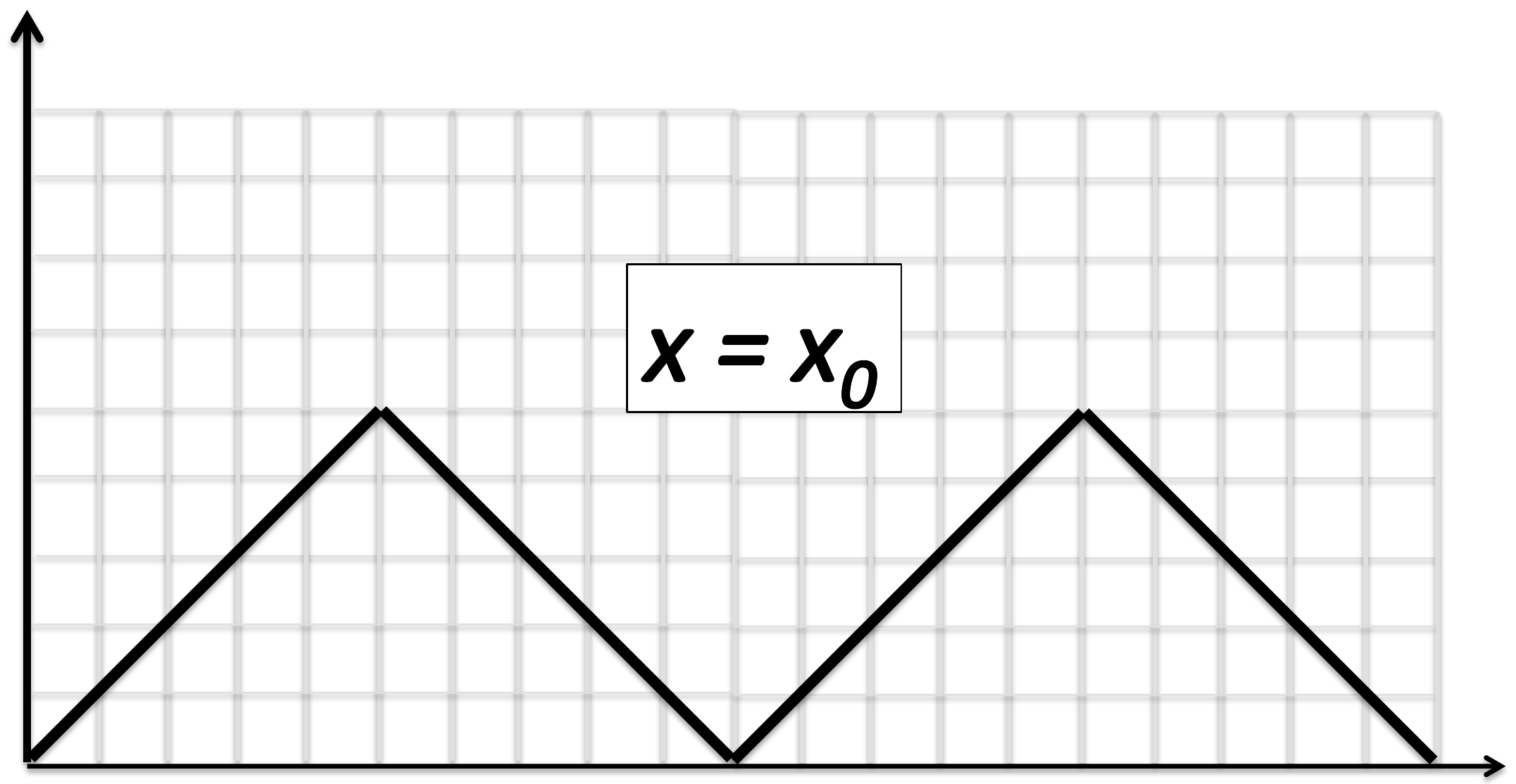}\includegraphics[scale=0.35]{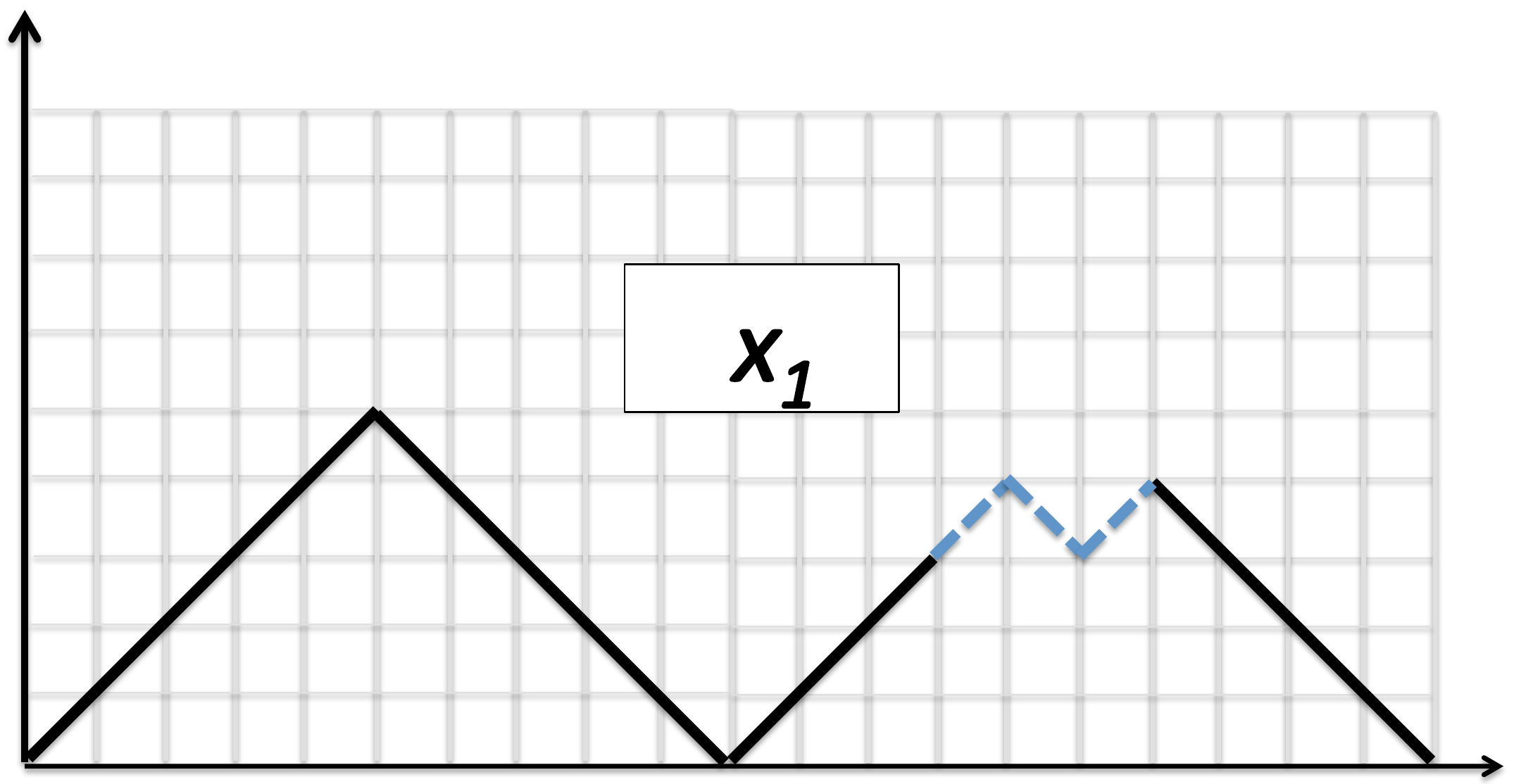}
\par\end{centering}
\begin{centering}
\includegraphics[scale=0.35]{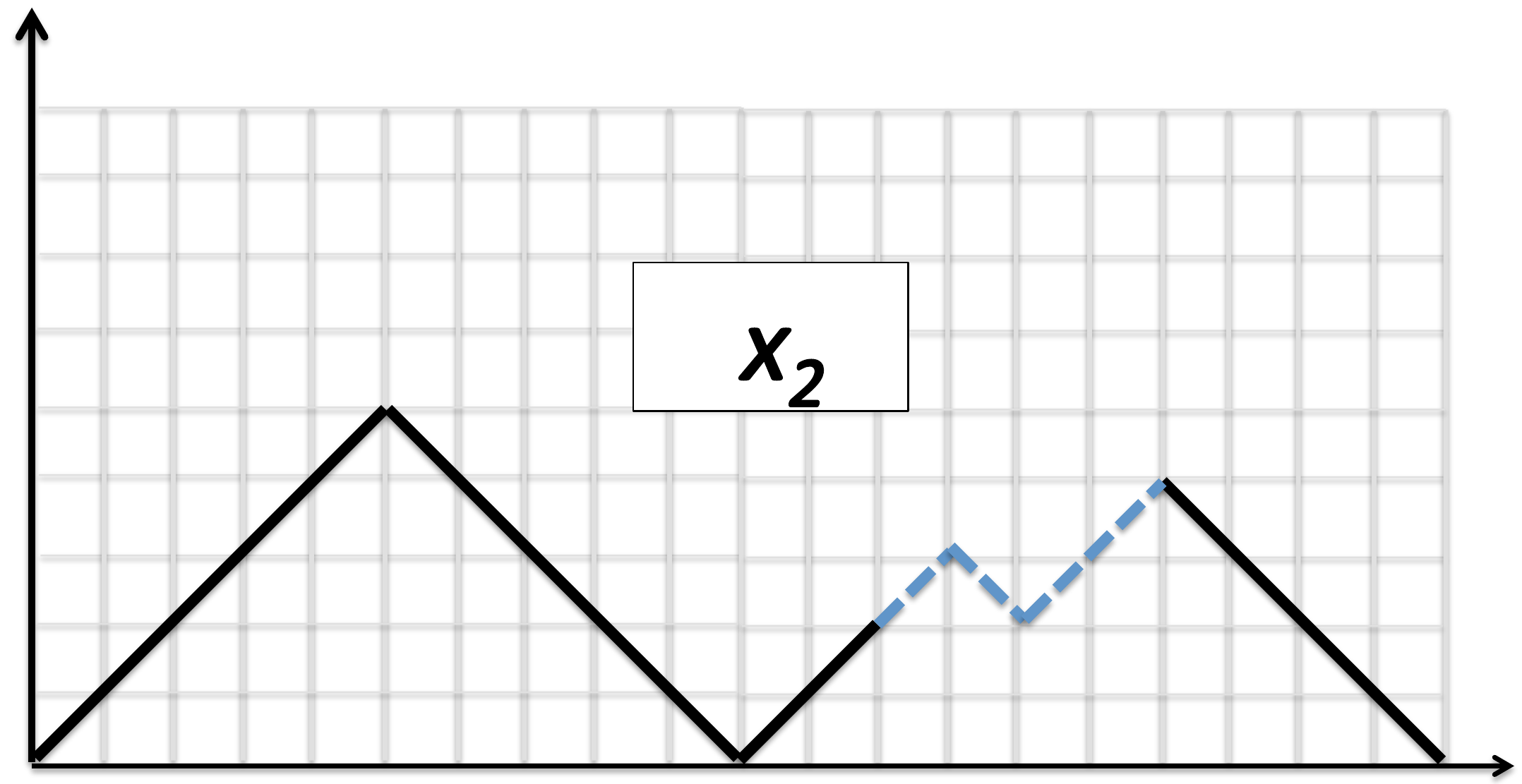}\includegraphics[scale=0.35]{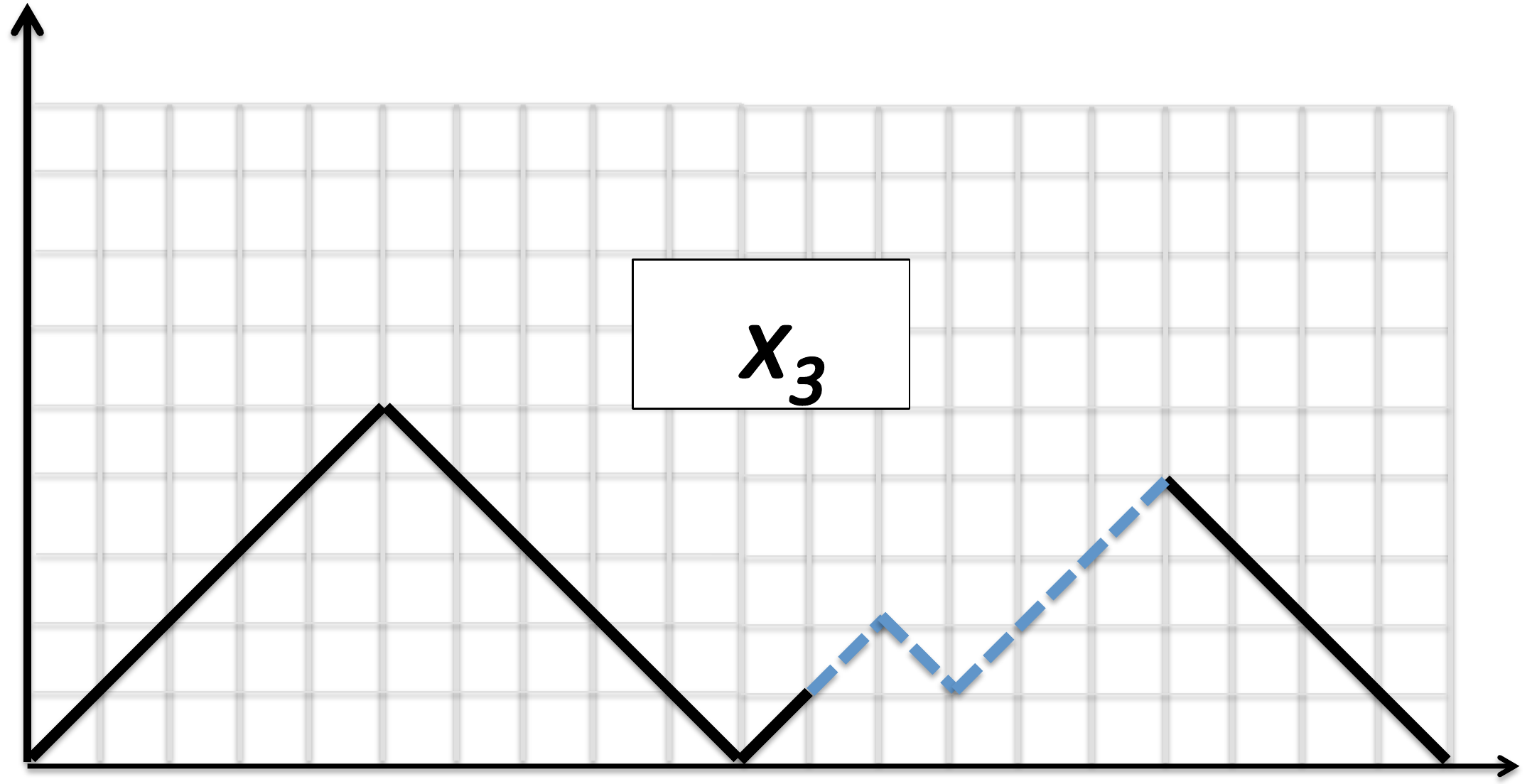}
\par\end{centering}
\begin{centering}
\includegraphics[scale=0.35]{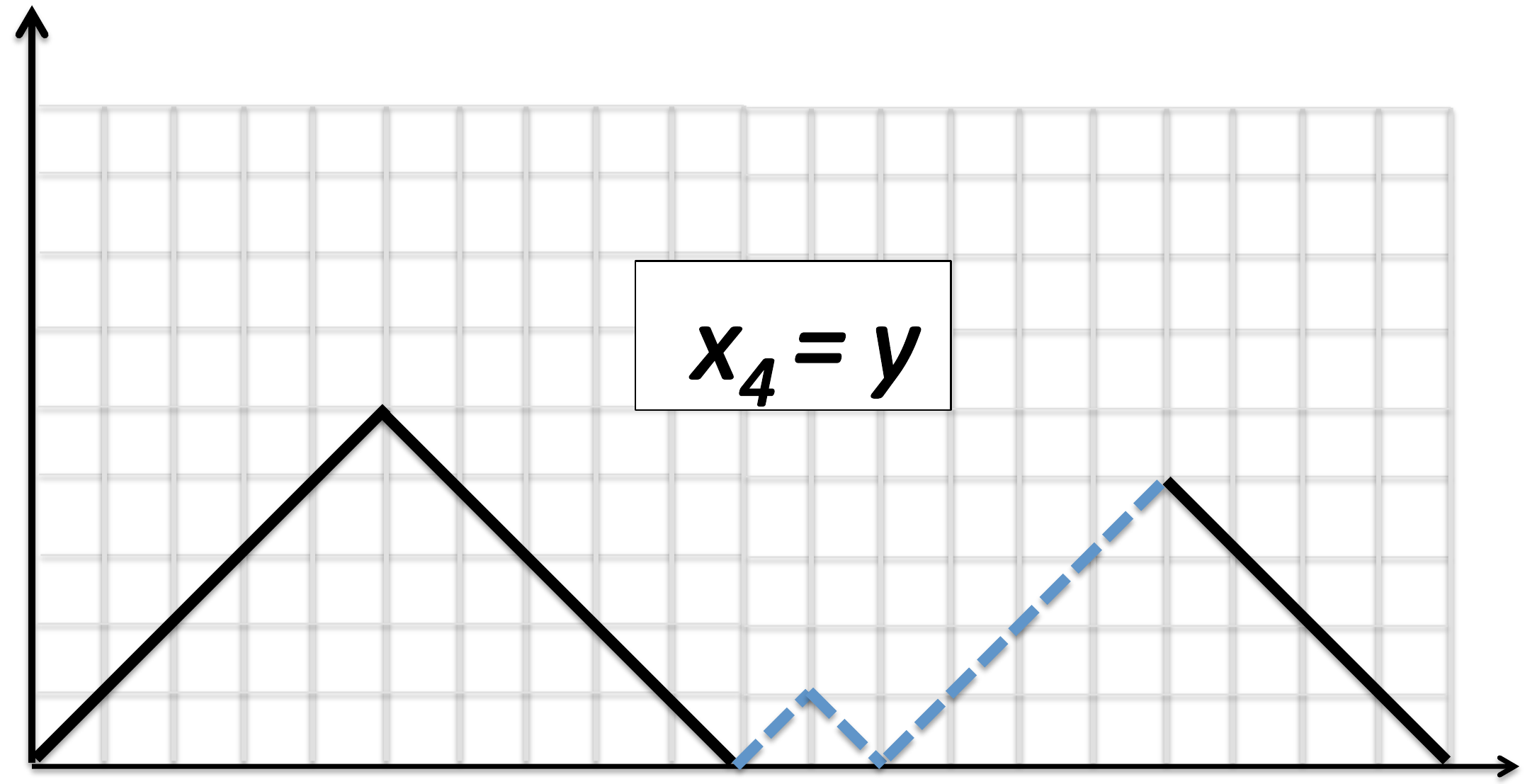}
\par\end{centering}
\caption{\label{fig:The-canonical-path}An example of canonical paths in Fredkin
Markov chain with the start state $x$ and final state $y$.}
\end{figure}
It remains to upper bound the quantity $A$. It is clear that the
length of any canonical path is at most $2n$. Since $\pi=\tilde{\pi}$,
and choosing a peak at random and inserting at a random position on
a chain of length $2n$ with $s$ possible choices of random coloring
corresponds to $\tilde{P}(x,y)\approx\frac{1}{s\left(2n\right)^{2}}$, Eq. \eqref{eq:A} is bounded by 
\[
A\le\max_{z\ne w:P(z,w)\ne0}\left\{ \frac{2n}{P(z,w)}\sum_{x\ne y:(z,w)\in\gamma_{x,y}}\tilde{P}(x,y)\right\} \le\frac{1}{2ns}\max_{z\ne w:P(z,w)\ne0}\frac{\sum_{x\ne y:(z,w)\in\gamma_{x,y}}1}{P(z,w)}
\]
The final step is to upper-bound the number of canonical paths going
through any edge of Fredkin Markov chain. Suppose $(z,w)$ is an edge
in the a canonical path that maximizes $A$. It is clear that
$z$ and $w$ are equal everywhere except from three consecutive positions
corresponding to exchange of a a peak with a single step. 

In any canonical
path that has $(z,w)$ as an edge, the peak could have started its
motion from, at most, $2n$ places on the path and similarly can terminate
its motion in at most $2n$ positions. This gives us $O(n^{2})$ paths that use the edge with the maximum edge load.  

Moreover, there can be at most $n$ such peaks whose motion would result in the same edge $(z,w)$ as we now show. Suppose the edge is the local move that connects $\mathsf{s_L}\mbox{ } udu  \mbox{ }\mathsf{s_R}$ to $\mathsf{s_L}\mbox{ } uud  \mbox{ }\mathsf{s_R}$, where $d$ and $u$ denote a step down and up, respectively, and  $\mathsf{s_L}$ and $\mathsf{s_R}$ denote a string on the left and a string on the right of the local move, respectively. In this configuration $\mathsf{s_L}$ and $\mathsf{s_R}$ count once as far as initial and final configurations are concerned; however, they can contribute $O(n)$ times in the number of paths that use this local move. This can be seen if we take, as an example,  $\mathsf{s_L}=udud...ud$; any of the peaks $ud$ could have started walking, leading to the local move. There can be $O(n)$ choices of such peaks in $\mathsf{s_L}$, all leading to the particular transition described just above.

Therefore, the total number of canonical paths using the maximum edge load is $O(n^{3})$. In the case that the peak can change
color at the last step, this upper bound becomes $O(n^{3}s)$. With
this we arrive at the upper-bound on $A$ 
\[
A\le\frac{O(n^{2})}{\text{min}_{z\ne w}P(z,w)}.
\]
Lastly we have the desired bound on the gap
\begin{equation}
1-\lambda_{1}(P)\ge\frac{s\text{ }\text{min}_{z\ne w}P(z,w)}{O(n^{15/2})}.\label{eq:Gap_Markov_Final}
\end{equation} \end{proof}
The gap gives us an upper-bound on the mixing time of the Fredkin
Markov chain via Eq. \eqref{eq:MixingTime}.

If we take $P(z,w)=1/2$, then we have a polynomial upper (lower)-bound
on the spectral gap (mixing time) of Fredkin Markov chain. In Section
\ref{subsec:Korepin} we define a suitable Markov chain in which $P(z,w)\propto1/n$,
which also yields a polynomial upper bound.

\textbf{5. An interesting failed attempt: Wilson's lattice path Markov
chain.} In \cite{wilson2004mixing}, David B. Wilson defines a \textit{lattice
path} from $(0,0)$ to $(2n,0$) to be a connected path with steps
$(1,1)$ and $(1,-1)$. In other words, a lattice path on $2n$ steps
is a walk from left to right between the coordinates $\left(x_{0},y_{0}\right)=\left(0,0\right)$
and $\left(x_{2n},y_{2n}\right)=\left(2n,0\right)$. And at each step
the $y-$coordinate changes by one $(x_{i+1},y_{i+1})=(x_{i}+1,y_{i}\pm1)$.
See Fig. \eqref{fig:A-lattice-path}. In particular, such lattice paths
can go below the $x-$axis 
\begin{figure}
\centering{}\includegraphics[scale=0.4]{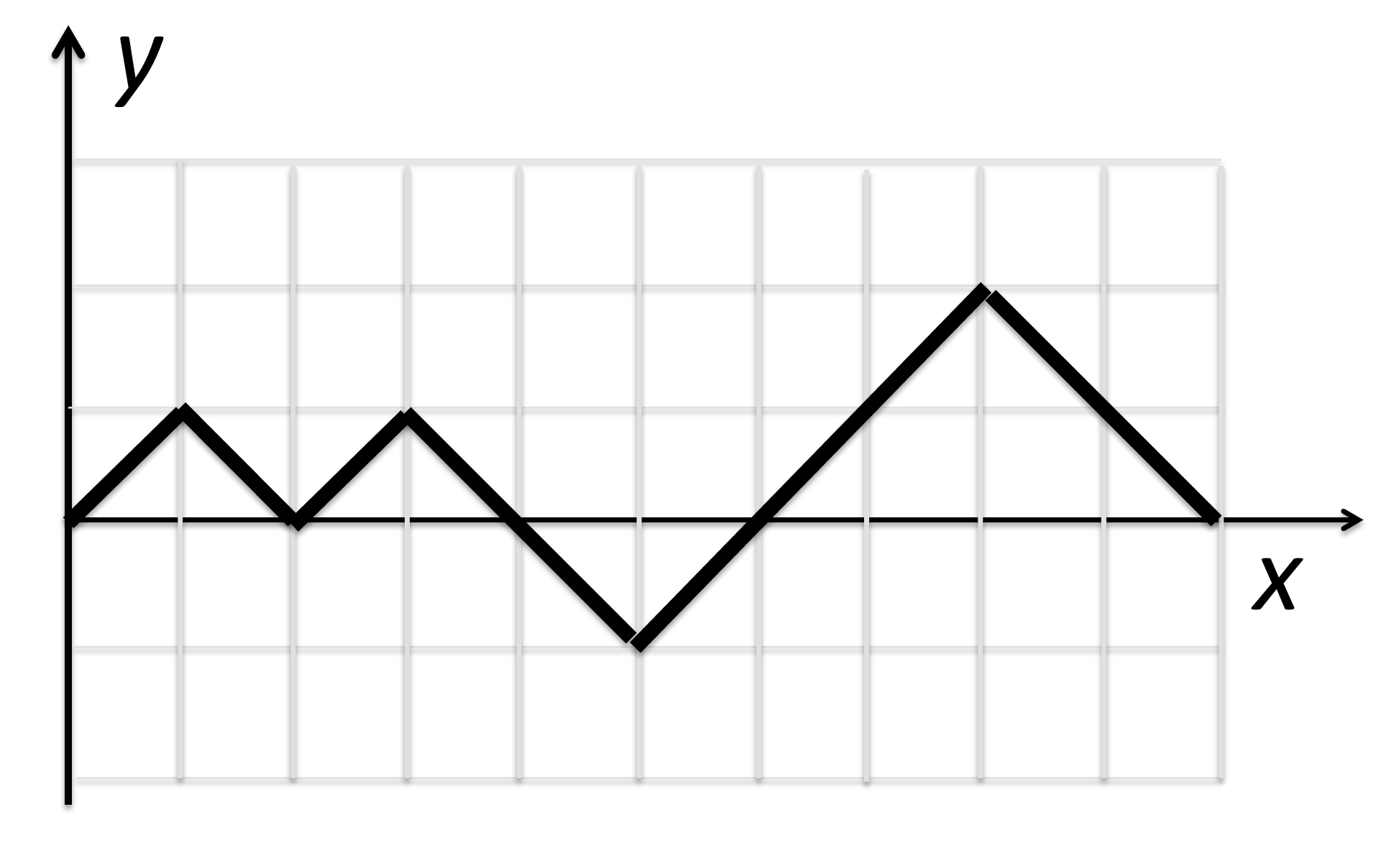}\caption{\label{fig:A-lattice-path}A lattice path through $5\times5$ square.}
\end{figure}

The \textit{lattice path Markov} \textit{chain} is local and has the
lattice paths as its state space. The transition rules are as follows:
Randomly pick two adjacent steps on the lattice path and if they correspond
to a peak then randomly change it to a valley and vice versa. The
mixing time of \textit{lattice path Markov} chain \cite{wilson2004mixing},
was proven by D.B. Wilson to be $\Theta(n^{3}\log n)$ \cite{wilson2004mixing}. It is noteworthy that this scaling is exact.

Let \textit{positive lattice path Markov} \textit{chain} be a local Markov chain, 
that has the Dyck paths as its state space. The transition rules are
the same as the lattice path Markov chain with the additional condition
that if changing a peak to a valley results in the walk crossing the
$x-$axis to become negative then the chain will skip that move in
that time step (i.e., the chain sits idle). The mixing time that
Wilson obtained also upper bounds the mixing time of the positive
lattice path Markov chain:
\begin{lem}
\label{lem:The-mixing-timeUpperBound}The mixing time of positive
lattice path Markov chain is upper bounded by $O(n^{3}\log n)$.
\end{lem}
\begin{proof}
The positive lattice path Markov chain has the space of Dyck walks
as its possible states, which is a subset of all possible lattice
paths. D. J. Aldous' Theorem says that if $X_{t}$ is a reversible
Markov chain on the state space $\Omega$ with the stationary measure
$\pi$ and spectral gap $\Delta$ and $B\subset\Omega$ is a non-empty
subset, then the gap of the chain induced on $B$ is $\Delta_{B}\ge\Delta$
(see also \cite[Theorem 13.20]{levin2009markov}). Therefore, the
spectral gap of lattice path Markov chain lower-bounds the gap of
positive lattice path chains. 
\end{proof}
So it seems quite natural that one would apply the comparison technique
to positive lattice path Markov chain and leverage Wilson's tight
bound to ultimately lower bound the spectral gap of the Fredkin Markov
chain. 
\begin{figure}
\centering{}\includegraphics[scale=0.4]{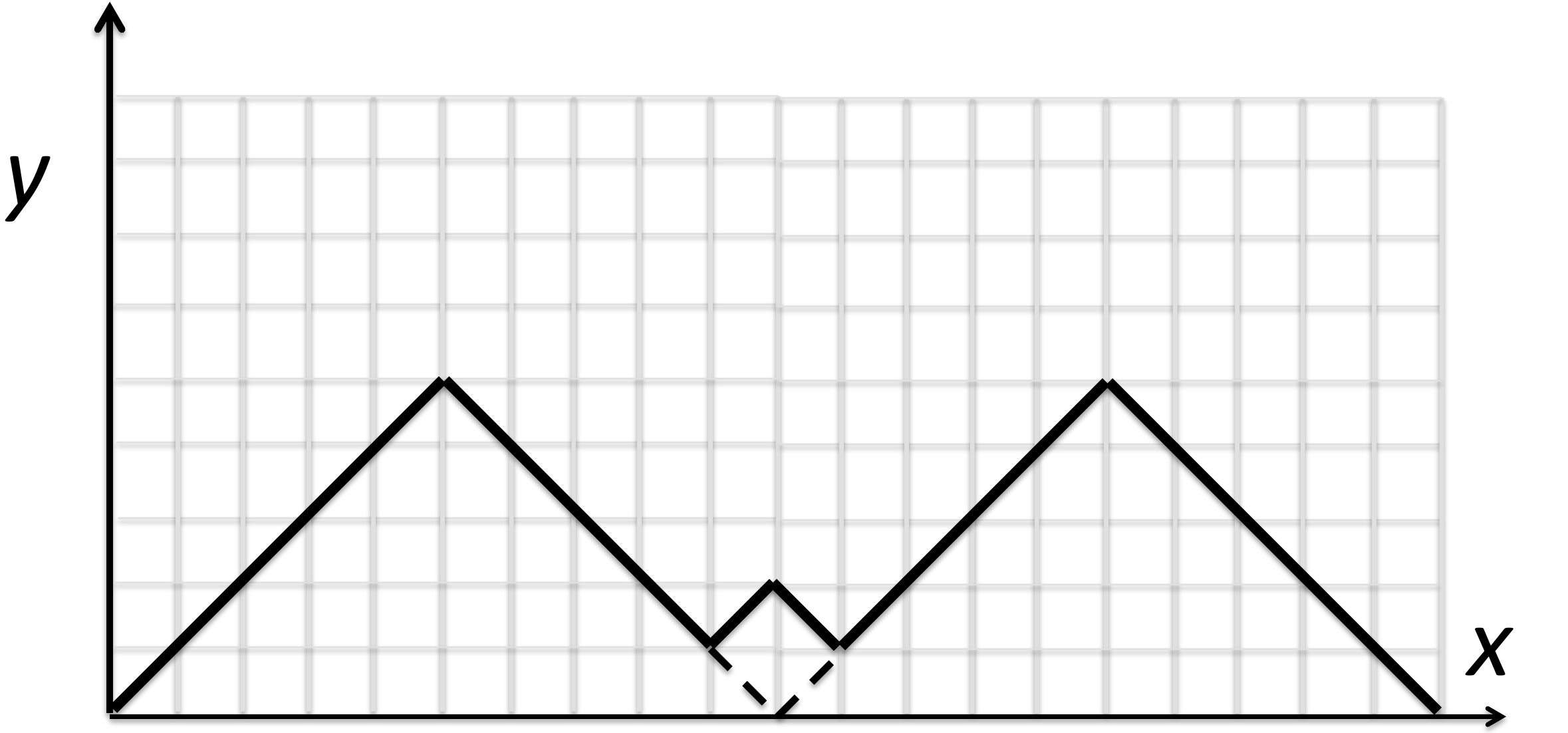}\caption{\label{fig:LongestPath}The local move between two adjacent states
in the lattice path Markov chain with no corresponding adjacent moves
in the Fredkin chain.}
\end{figure}
Although the Fredkin Markov chain is similar to the positive lattice
path Markov chain, it has a crucial difference. 

For every Fredkin adjacent move there is a corresponding adjacent move in positive lattice
paths. That is, there is a flipping of a peak or valley that relates
any two adjacent states under Fredkin Markov chain. In Fig. \eqref{fig:localMoves-1} imagine flipping peaks to valleys and vice versa, then the resulting local moves coincide with Fredkin moves. 

However, the converse is not true; there are adjacent states in the positive Markov chain
that do not have equivalent adjacent states under the Fredkin Markov
chain. In particular, any valley with side lengths of at least $2$
 has an adjacent state with a single peak at the bottom of
the dale under positive Markov chain moves. But the Fredkin Markov chain
cannot relate them with a single move. See Fig. \eqref{fig:LongestPath}
for an example of two adjacent states under a positive lattice path
Markov chain with no single Fredkin local move that would relate them.

Although we could find paths of length $O(n)$ to relate the worst case scenario shown in Fig. \eqref{fig:LongestPath}, the bounds on $A$ ended up being exponentially large in $n$.  This was because exponentially many canonical paths would end up needing to use a single edge. Therefore, in Section \ref{subsec:Korepin} we apply the comparison technique using a different Markov chain that gives us the desired result. 

It will be very exciting if, like the lattice path Markov chain, an exact bound is obtained for the Fredkin Markov chain. We leave this problem for future work.
\section{\label{sec:Hamiltonian-Gap-of}Hamiltonian Gap of Recent Exactly
Solvable Models}
\subsection{\label{subsec:Mapping_Markov_SpinChain}Mapping quantum spin Hamiltonians
to classical Markov chains}
There is a fascinating bridge between certain classical and quantum
 statistical mechanical systems. This connection goes beyond the well-known
equivalence of certain $D$ dimensional quantum systems and $(D+1)$
dimensional classical systems \cite{polyakov1987gauge}. Suppose we
have a quantum Hamiltonian 
\begin{equation}
H=\sum_{j}H_{j}.\label{eq:H_general}
\end{equation}
\begin{defn}
$H$ is called $k-$local if each $H_{j}$ acts non-trivially only
on a constant number ($k)$ of particles (e.g, spins). For example
a spin chain with  nearest neighbors interaction is a $2-$local
Hamiltonian.
\end{defn}
\begin{defn}
$H$ is \textit{Frustration Free (FF)} if the ground state of $H$,
which is the eigenvector corresponding to the lowest eigenvalue (ground
state energy) of $H$, is also the ground state of every $H_{j}$. 
\end{defn}
\begin{defn}
\label{Def:Stoch}$H$ is called \textit{Stoquastic} with respect
to a basis ${\cal B}$ if each $H_{j}$ has non-positive off-diagonal
entries with respect to ${\cal B}$ \cite{bravyi2006complexity}. 
\end{defn}
\begin{defn}
The gap of a quantum Hamiltonian $H$ is the positive difference of
its two {\it smallest} eigenvalues and is denoted by $\Delta(H)$. 
\end{defn}
Comment: The definition of the gap is different from the definition
in the Markov chain literature, where the gap is the difference of
the two \textit{largest} eigenvalues.

Comment: The usage of Perron-Frobenius theorem in the context of quantum spin systems has a long history. In particular the non-positivity was utilized in Marshall's paper \cite{marshall1955antiferromagnetism} and Mattis-Lieb theory \cite{lieb1962ordering} as well as more recent  works \cite{aizenman1994geometric}. In particular, the stability of the gap along a path that connects the Hamiltonian to an exactly solvable Hamiltonian is a result of positivity. 

The term stoquastic is meant to emphasize the connection between classical
stochastic matrices and quantum Hamiltonians. Stoquastic Hamiltonians
avoid the so called sign problem and their ground state properties
can be simulated using classical Monte Carlo algorithms \cite[and references therein.]{bravyi2009complexity}.
With no loss of generality the smallest eigenvalue of any FF Hamiltonian
is taken to be zero and we have $H|\psi\rangle=0$, where $|\psi\rangle$
is the ground state of $H$ and FF condition implies that $H_{j}|\psi\rangle=0$
for all $j$. 
\begin{rem}
If we multiply the generator or generating matrix of continuous time
Markov chains by minus one, the resulting matrix would be a stoquastic
matrix. However, stoquastic matrices may have diagonal entries of
any sign. 
\end{rem}
Using the Perron-Frobenius theorem, Bravyi and Terhal showed that the
ground state $|\psi\rangle$ of any stoquastic FF Hamiltonian can
be chosen to be a vector with non-negative amplitudes on a standard
basis \cite{bravyi2009complexity}. Suppose $H$ is restricted to the ground subspace. Following \cite{bravyi2009complexity}
let us define the transition matrix 
\begin{equation}
P(x,y)=\delta_{x,y}-\beta\sqrt{\frac{\pi(y)}{\pi(x)}}\langle x|H|y\rangle\label{eq:P_General}
\end{equation}
where $\beta>0$ is real and chosen such that $P(x,y)\ge0$ and $\pi(x)=\langle x|\psi\rangle^{2}$
is the stationary distribution. The eigenvalue equation $(\mathbb{I}-\beta H)|\psi\rangle=|\psi\rangle$
implies that $\sum_{y}P(x,y)=1$. Therefore the matrix $P$ really
defines a random walk. 

Recall that the spectral gap of a Markov chain is the difference of
its two largest eigenvalues, i.e. $1-\lambda_{1}(P)$. Because of
the minus sign in Eq. \eqref{eq:P_General} the energy gap of the FF
Hamiltonian is related to the gap of the Markov chain by
\[
\Delta(H)=\frac{\lambda_{1}(P)}{\beta}.
\]

This provides a powerful tool for quantifying the gap of a class of
quantum local Hamiltonians using the arsenal of existing techniques
for proving the gap of classical Markov chains. 

We will use this connection along with the scaling of the gap proved
for the Fredkin Markov chain in the previous section to prove the
gap of the Fredkin \textit{quantum spin} chain model \cite{salberger2016fredkin,dell2016violation}. 

\subsection{\label{subsec:Movassagh_shor}The Motzkin spin chain and its ground
state}
\begin{figure}
\begin{centering}
\includegraphics[scale=0.3]{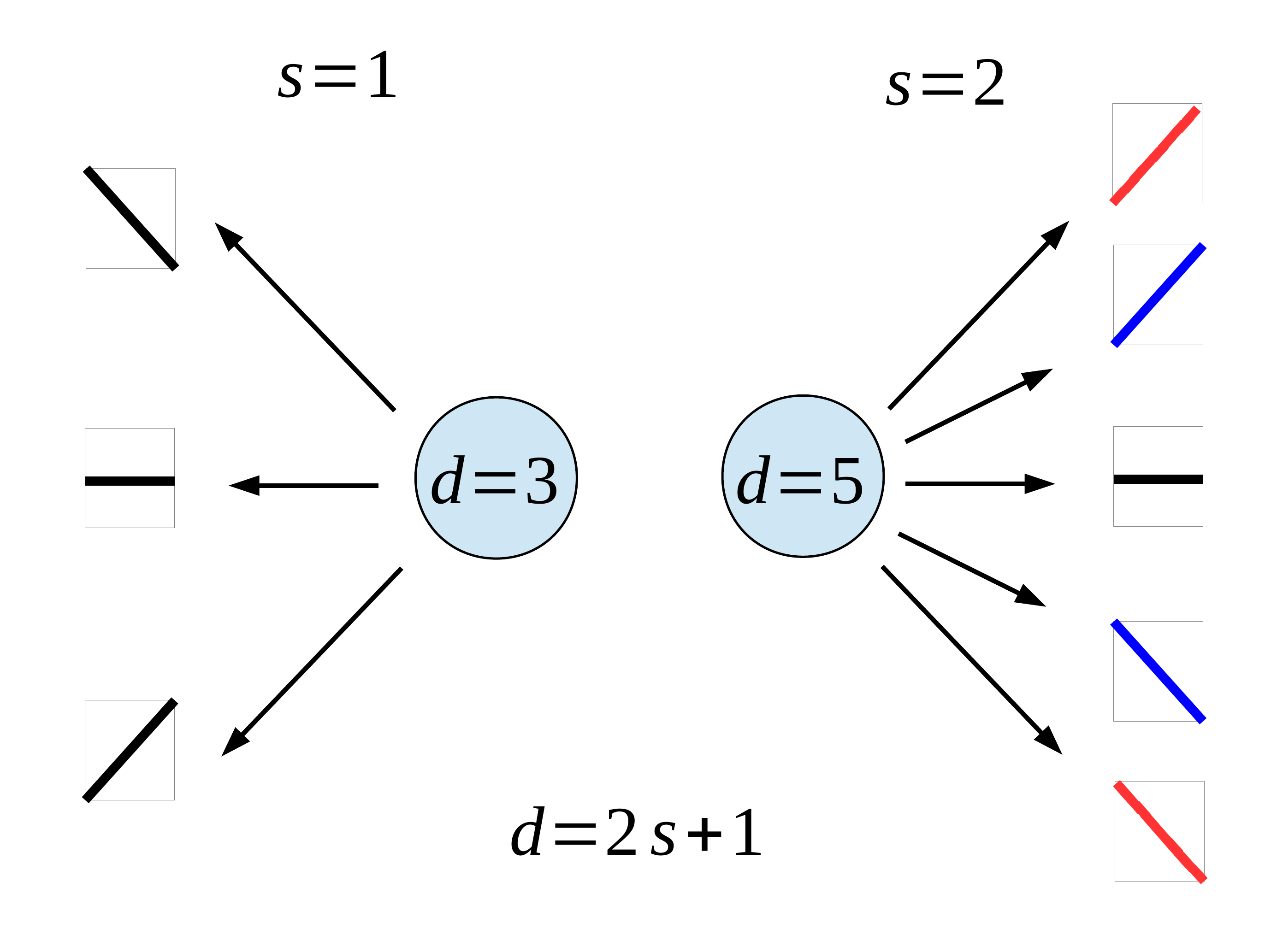}\caption{\label{fig:spin_States}Labels for the $2s+1$ states for $s=1$ and
$s=2$. Note that the flat steps are always black for all $s$. }
\par\end{centering}
\end{figure}
The predecessor of the Fredkin spin chain is the colored Motzkin spin chain \cite{movassagh2016supercritical},
which we now describe. Throughout we take the length of the chain
to be $2n$. Let us consider an integer spin$-s$ chain of length
$2n$. It is convenient to label the $d=2s+1$ spin states by steps
up and down of $s$ different colors as shown in Fig. \eqref{fig:spin_States}.
Equivalently, and for better readability, we instead use the labels
$\left\{u^1,u^2,\cdots,u^s,0,d^1,d^2,\cdots,d^s\right\}$ where $u$
means a step up and $d$ a step down. We distinguish each \textit{type}
of step by associating a color from the $s$ colors shown as superscripts
on $u$ and $d$. 

A Motzkin walk on $2n$ steps is any walk from $\left(x,y\right)=\left(0,0\right)$
to $\left(x,y\right)=\left(2n,0\right)$ with steps $\left(1,0\right)$,
$\left(1,1\right)$ and $\left(1,-1\right)$ that never passes below
the $x-$axis, i.e., $y\ge0$. An example of such a walk is shown
in Fig. \eqref{fig:Motzkin}. In our model the unique ground state is
the $s-$colored \textit{Motzkin state} which is defined to be the
uniform superposition of all $s$ colorings of Motzkin walks on $2n$
steps. 

The Schmidt rank is $\frac{s^{n+1}-1}{s-1}\approx\frac{s^{n+1}}{s-1}$,
and the half-chain entanglement entropy asymptotically is \cite{movassagh2016supercritical}
\begin{eqnarray*}
S & = & 2\log_{2}\left(s\right)\mbox{ }\sqrt{\frac{2\sigma}{\pi}}\mbox{ }\sqrt{n}+\frac{1}{2}\log_{2}\left(2\pi\sigma n\right)+\left(\gamma-\frac{1}{2}\right)\log_{2}e\quad\mbox{bits}
\end{eqnarray*}
where $\sigma=\frac{\sqrt{s}}{2\sqrt{s}+1}$ and $\gamma$ is Euler's
constant. The ground state is a pure state (the Motzkin state), whose
entanglement entropy is zero. However, the entanglement entropy quantifies
the amount of disorder produced (i.e., information lost) by ignoring
a subset of the chain. The leading order $\sqrt{n}$ scaling of the
half-chain entropy establishes that there is a large amount of quantum
correlations between the two halves. 
\begin{figure}
\begin{centering}
\includegraphics[scale=0.35]{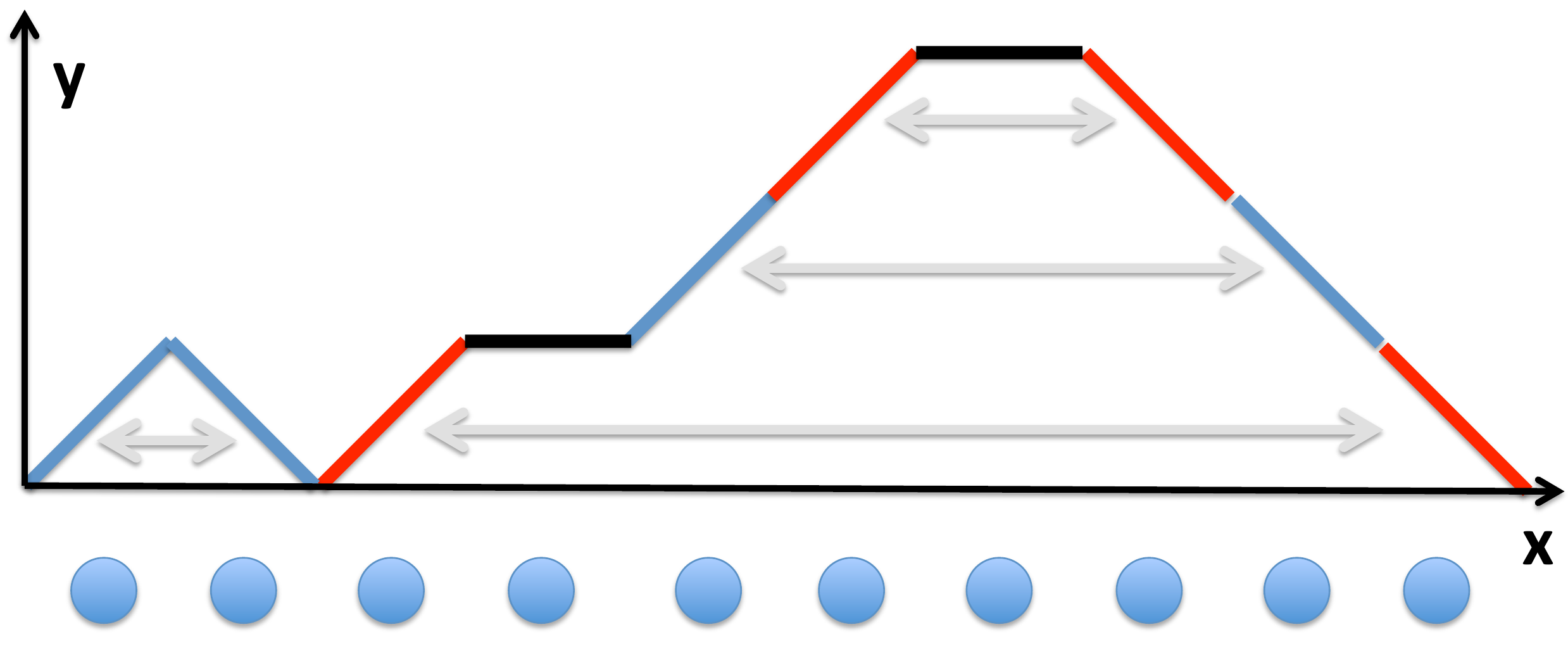}
\par\end{centering}
\centering{}\caption{\label{fig:Motzkin}A Motzkin walk with $s=2$ colors on a chain of
length $2n=10$. }
\end{figure}
Consider the following local operations to any Motzkin walk: interchanging
zero with a non-flat step (i.e., $0u^{k}\leftrightarrow u^{k}0$ or
$0d^{k}\leftrightarrow d^{k}0$) or interchanging a consecutive pair
of zeros with a peak of a given color (i.e., $00\leftrightarrow u^{k}d^{k}$).
These are shown in Fig. \eqref{fig:Local-moves-Motzkin}. Any $s-$colored
Motzkin walk can be obtained from another one by a sequence of these
local changes. 

To construct a local Hamiltonian with projectors as interactions that
has the uniform superposition of the Motzkin walks as its zero energy
ground state, each of the local terms of the Hamiltonian has to annihilate
states that are symmetric under these interchanges. Physically, local projectors
as interactions have the advantage of being robust against certain
perturbations \cite{Kraus2008}. This is important from a practical
point of view and experimental realizations.
\begin{figure}
\begin{centering}
\includegraphics[scale=0.45]{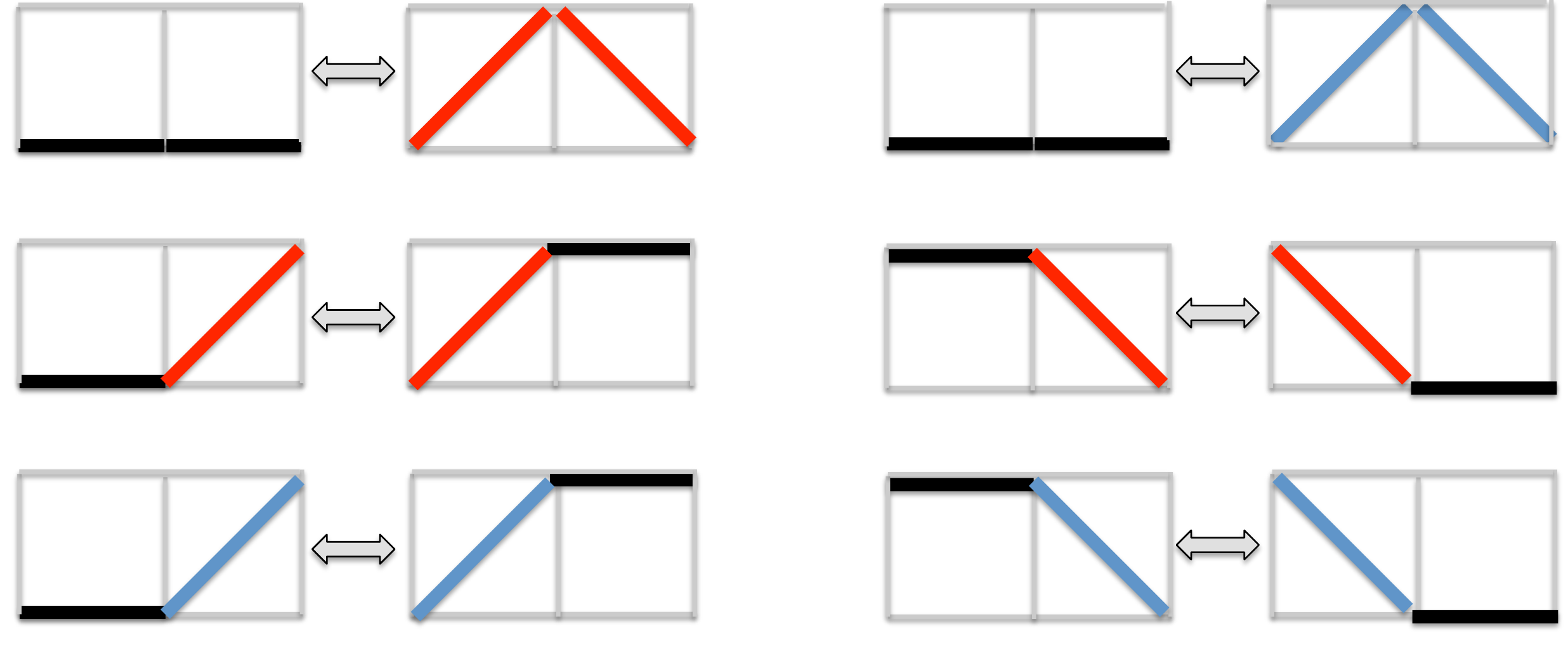}\caption{\label{fig:Local-moves-Motzkin}Local moves for $s=2$.}
\par\end{centering}
\end{figure}
The local Hamiltonian that has the Motzkin state as its unique zero
energy ground state is \cite{movassagh2016supercritical}
\begin{equation}
H=\Pi_{boundary}+\sum_{j=1}^{2n-1}\Pi_{j,j+1}+\sum_{j=1}^{2n-1}\Pi_{j,j+1}^{cross},\label{eq:H}
\end{equation}
where $\Pi_{j,j+1}$ implements the local operations discussed above
and is defined by 
\[
\Pi_{j,j+1}\equiv\sum_{k=1}^{s}\left[|U^{k}\rangle_{j,j+1}\langle U{}^{k}|+|D^{k}\rangle_{j,j+1}\langle D^{k}|+|\varphi^{k}\rangle_{j,j+1}\langle\varphi^{k}|\right]
\]
with $|U^{k}\rangle=\frac{1}{\sqrt{2}}\left[|0u^{k}\rangle-|u^{k}0\rangle\right]$,
$|D^{k}\rangle=\frac{1}{\sqrt{2}}\left[|0d^{k}\rangle-|d^{k}0\rangle\right]$
and $|\varphi^{k}\rangle=\frac{1}{\sqrt{2}}\left[|00\rangle-|u^{k}d^{k}\rangle\right]$.
The term $\Pi_{boundary}\equiv\sum_{k=1}^{s}\left[|d^{k}\rangle_{1}\langle d^{k}|+|u^{k}\rangle_{2n}\langle u^{k}|\right]$
selects out the Motzkin state by excluding all walks that start and
end at non-zero heights. Lastly, $\Pi_{j,j+1}^{cross}=\sum_{k\ne i}|u^{k}d^{i}\rangle_{j,j+1}\langle u^{k}d^{i}|$
ensures that balancing is well ordered (i.e., prohibits $00\leftrightarrow u^{k}d^{i}$);
these projectors are required only when $s>1$ and do not appear in
\cite{Movassagh2012_brackets}. 

In \cite{movassagh2016supercritical} we proved that the energy gap to the
first excited state is $\Theta(n^{-c})$. 
\subsection{\label{subsec:Korepin}The Fredkin spin chain and its energy gap}
The Fredkin spin chain is a Fermionic extension of \cite{movassagh2016supercritical}
in which $d$ is even. The spin states are naturally labeled as shown
in See Fig. \eqref{fig:Labeling-of-the}.
\begin{figure}
\begin{centering}
\includegraphics[scale=0.45]{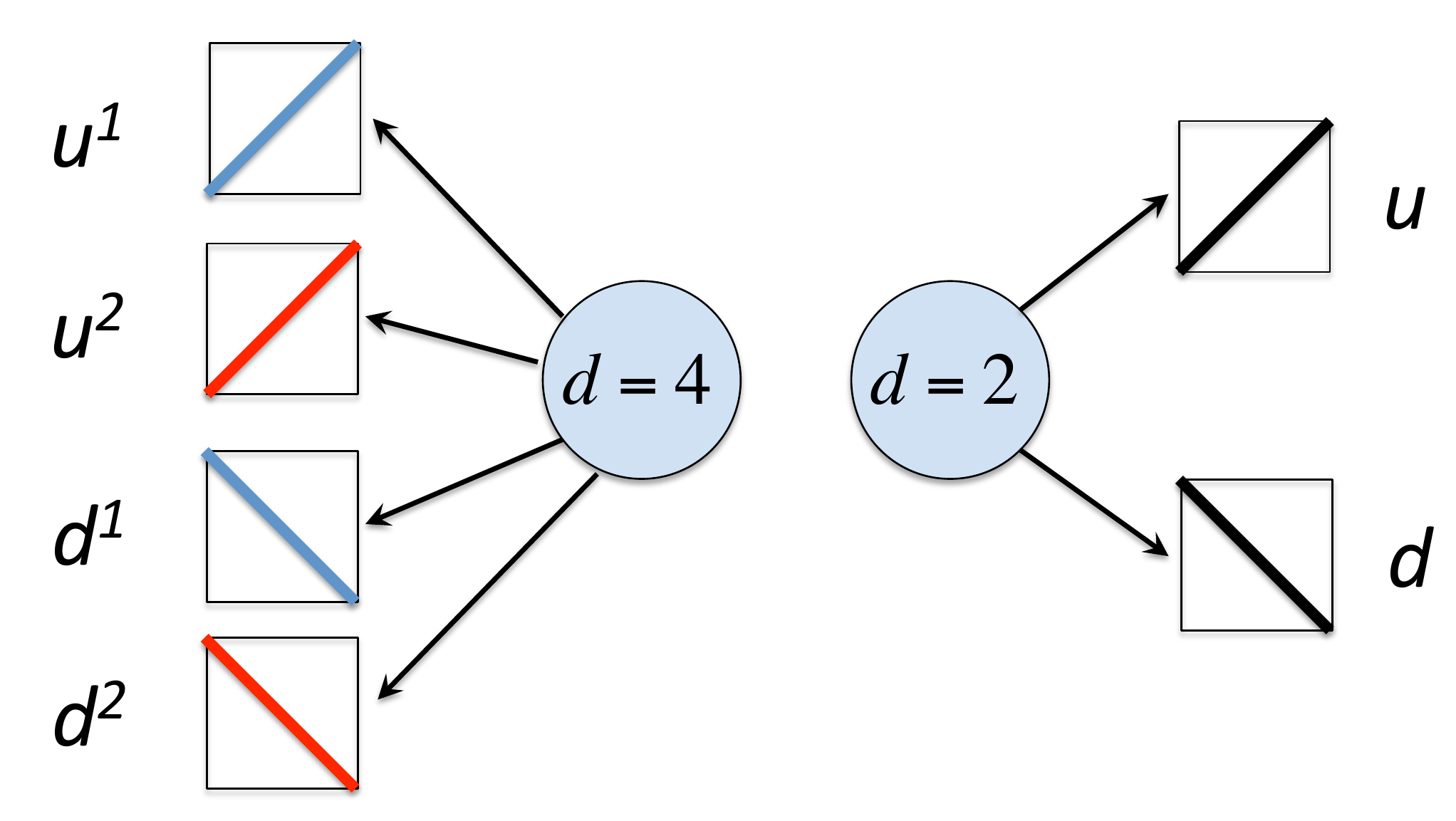}
\par\end{centering}
\caption{\label{fig:Labeling-of-the}Labeling of the spin states for half-integer
spin $s$. Left $s=3/2$ and on right $s=1/2$. }
\end{figure}

In \cite{salberger2016fredkin}, the unique ground state is the $s-$colored
\textit{Dyck state} which is defined to be the uniform superposition
of all $s$ colored Dyck paths on $2n$ steps. An example of such
a path was shown in Fig. \eqref{fig:Dyck}. 

Therefore, the Fredkin spin chain's unique ground state, denoted by
$|{\cal D}^{s}\rangle$, is 
\[
|{\cal D}^{s}\rangle=\frac{1}{\sqrt{N}}\sum_{\begin{array}{c}
w\in\{s-\text{colored }\\
\text{Dyck walks}\}
\end{array}}|w\rangle
\]
where $N=s^{n}C_{n}$ is the normalization , with $C_{n}$ denoting the $n^{th}$ Catalan number. The Hamiltonian whose zero
energy state is $|{\cal D}^{s}\rangle$ is 3-local and reads
\begin{equation}
H=\sum_{j=1}^{2n-2}\Pi_{j,j+1,j+2}+\sum_{j=1}^{2n-1}\Pi_{j,j+1}^{cross}+\Pi_{boundary}\label{eq:HKorepin}
\end{equation}
where 
\begin{eqnarray*}
\Pi_{j,j+1,j+2} & = & \sum_{k_{1},k_{2}=1}^{s}|U^{k_{1},k_{2}}\rangle_{j,j+1,j+2}\langle U^{k_{1},k_{2}}|+|D^{k_{1},k_{2}}\rangle_{j,j+1,j+2}\langle D^{k_{1},k_{2}}|+|\varphi^{k_{1},k_{2}}\rangle_{j,j+1,j+2}\langle\varphi^{k_{1},k_{2}}|\\
\Pi_{j,j+1}^{cross} & = & \sum_{k_{1}\ne k_{2}=1}^{s}|u^{k_{1}}d^{k_{2}}\rangle_{j,j+1}\langle u^{k_{1}}d^{k_{2}}|\\
\Pi_{boundary} & = & \sum_{k=1}^{s}|d^{k}\rangle_{1}\langle d^{k}|+|u^{k}\rangle_{2n}\langle u^{k}|
\end{eqnarray*}
and $|U^{k_{1},k_{2}}\rangle=\frac{1}{\sqrt{2}}\left\{ |u^{k_{1}}u^{k_{2}}d^{k_{2}}\rangle-|u^{k_{2}}d^{k_{2}}u^{k_{1}}\rangle\right\} $
, $|D^{k_{1},k_{2}}\rangle=\frac{1}{\sqrt{2}}\left\{ |d^{k_{1}}u^{k_{2}}d^{k_{2}}\rangle-|u^{k_{2}}d^{k_{2}}d^{k_{1}}\rangle\right\} $
and $|\varphi^{k_{1},k_{2}}\rangle=\frac{1}{\sqrt{2}}\left\{ |u^{k_{1}}d^{k_{1}}\rangle-|u^{k_{2}}d^{k_{2}}\rangle\right\} $,
where $k_{1}$ and $k_{2}$ denote the colors of each step. 

$\Pi_{j,j+1,j+2}$ in the Hamiltonian implements the Glauber dynamics
whose local moves are shown in Fig. \eqref{fig:localMoves_Colored}.
The projectors $|U^{k_{1},k_{2}}\rangle\langle U^{k_{1},k_{2}}|$
and $|D^{k_{1},k_{2}}\rangle\langle D^{k_{1},k_{2}}|$ implement the
exchange of a peak with a step up or down and $|\varphi^{k_{1},k_{2}}\rangle\langle\varphi^{k_{1},k_{2}}|$
implements the recoloring of a peak (compare with Fig. \eqref{fig:localMoves_Colored}).
$\sum_{j}\Pi_{j,j+1}^{cross}$ ensures that the matching steps up
and down have the same color. Lastly, $\Pi_{boundary}$ ensures that
the walk always starts and ends at height zero. 
\begin{thm*}\label{MainTheorem}
The gap of Fredkin spin chain is $\Delta(H)=O(n^{-c})$, with $c\ge 2$.
\end{thm*}
Below we prove this theorem by the following strategy:

1. Upper-bound: Choose a state $|\phi\rangle$ which is a uniform
superposition of Dyck paths except that each path has a complex phase
whose exponent is proportional to the area between the Dyck path and
the $x-$axis. We then utilize ideas from universality of Brownian
motion and Brownian excursions to prove an $O(n^{-2})$ upper bound
on the gap of the Hamiltonian.

2. Lower-bound in the balanced subspace: We restrict the Hamiltonian
to the balanced subspace which is the subspace spanned by paths that
start at zero height and end at zero height and never become negative.
We prove that the gap in this subspace is $\Delta(H)\ge O(sn^{-15/2})$.
We obtain this by casting the original Hamiltonian onto a classical
Markov chain and then use the results of Section \ref{sec:New-Results-MC}.

3. Lower-bound in the unbalanced subspace: Prove that the lowest energy
in the unbalanced subspace is polynomially small. We achieve this
by expressing the Hamiltonian as an effective next-nearest neighbors
hopping matrix. We then solve the ground state of the effective Hopping matrix exactly. \\
\subsubsection{The upper-bound on the gap is $O(n^{-c})$, $c\ge2$}
Let $s=1$ for now and $|\phi\rangle$ be any state with a constant overlap with the zero energy
ground state $|{\cal D}\rangle$. Let $|\langle\phi|D\rangle|^{2}\le1/2$; we now show that $\langle\phi|H|\phi\rangle\ge\frac{1}{2}\Delta(H)$.
First the test state can be written in the orthonormal eigenbasis
of the Hamiltonian via
\[
|\phi\rangle=\alpha_{0}|{\cal D}\rangle+\alpha_{1}|e_{1}\rangle+\dots+\alpha_{(C_{n}-1)}|e_{(C_{n}-1)}\rangle,
\]
where $|e_{1}\rangle,\dots,|e_{(C_{n}-1)}\rangle$ are the excited
states of $H$. Since $|\langle\phi|D\rangle|^{2}\le1/2$ , we have
\[
\langle\phi|H|\phi\rangle=\sum_{i\ge1}|\alpha_{i}|^{2}\langle e_{i}|H|e_{i}\rangle\ge\langle e_{1}|H|e_{1}\rangle\sum_{i\ge1}|\alpha_{i}|^{2}\ge\frac{1}{2}\langle e_{1}|H|e_{1}\rangle,
\]
where  $\Delta(H)=\langle e_{1}|H|e_{1}\rangle$ is the energy gap. 
Therefore, all we need is a good test state $|\phi\rangle$ with a
small overlap with the ground state for which we can prove $\langle\phi|H|\phi\rangle\ge O(n^{-c})$.
\begin{lem}\label{Lem:UpperBound}
The gap of Fredkin spin chain, $\Delta(H)\ge O(n^{-2})$.
\end{lem}
\begin{proof}
Let $|\phi\rangle=\frac{1}{\sqrt{C_{n}}}\sum_{\mathsf{s}\in D}e^{2\pi i\tilde{A}_{s}\tilde{\theta}}|\mathsf{s}\rangle$,
and we have 
\[
\langle\phi|D\rangle=\frac{1}{C_{n}}\sum_{\mathsf{s}\in D}e^{2\pi i\tilde{A}_{s}\tilde{\theta}}
\]
as $n\rightarrow\infty$, the random walk converges to a Wiener process
and a random Dyck walk to a Brownian excursion. Currently the length
of the walk is over $2n$ steps. We need to map this walk to the interval
$[0,1]$ and do so by scaling the walks such that they take the standard
form.

Recall that a Brownian excursion, $B(t)$, is a Brownian motion that
takes place on the unit interval, starts and ends at height zero,
and never becomes negative. That is, $B(t)$ is a Brownian motion
on the $t\in[0,1]$, with $B(0)=B(1)=0$ and $B(t)\ge0$. Let the
area below the Brownian excursion be \cite{janson2007brownian} 
\[
B_{ex}=\int_{0}^{1}B(t)\mbox{ }dt.
\]

Let $f_{A}(x)$ be the probability density function of $B_{ex}$. 
It is known that \cite[p. 92]{janson2007brownian}
\[
f_{A}(x)=\frac{2\sqrt{6}}{x^{2}}\sum_{j=1}^{\infty}v_{j}^{2/3}e^{-v_{j}}U(-\frac{5}{4},\frac{4}{3};v_{j})\qquad x\in[0,\infty)
\]
with $v_{j}=2|a_{j}|^{3}/27x^{2}$ where $a_{j}$ are the zeros of
the Airy function, $Ai(x)$, and $U$ is the confluent hypergeometric
function. The expected value and standard of deviation of $B_{ex}$, respectively, 
are 
\begin{eqnarray*}
\mathbb{E}[B_{ex}] & = & \frac{1}{2}\sqrt{\frac{\pi}{2}}\\
\sigma_{B_{ex}} & = & \sqrt{5/12-\pi/8}.
\end{eqnarray*}

The sum of the areas of the Dyck paths of length $2n$ is \cite{chapman1999moments}
\[
A_{2n}=4^{n}-\frac{1}{2}\left(\begin{array}{c}
2n+2\\
n+1
\end{array}\right)\sim2\mbox{ }\frac{4^{n}}{\sqrt{\pi(n+1)}}.
\]
Since $C_{n}\sim\frac{4^{n}}{n^{3/2}\sqrt{\pi}}$, the expected area
is 
\[
\mathbb{E}[\tilde{A}_{s}]=\frac{A_{2n}}{C_{n}}\sim\sqrt{\pi}n^{3/2}.
\]

We now define the scaling constant, $c$, by 
\[
\mathbb{E}[\tilde{A}_{s}]=cn^{3/2}\mathbb{E}[B_{ex}].
\]
This gives $c=2\sqrt{2}$. Make the following change of variables $\tilde{A}_{p}=2\sqrt{2}n^{3/2}x$
and $\tilde{\theta}=\frac{n^{-3/2}}{2\sqrt{2}}\theta$. With these
scalings we have mapped the problem onto a standard Brownian excursion. 

The overlap, in the limit, is the characteristic function of the density
of the area under the excursion 
\[
\lim_{n\rightarrow\infty}\langle{\cal D}|\phi\rangle\approx F_{A}(\theta)\equiv\int_{0}^{\infty}f_{A}(x)e^{2\pi ix\theta}.
\]

Taking $\theta=1/\sigma$ with $\sigma$ being the standard of deviation, as
done in \cite{movassagh2016supercritical}, we arrive at $\tilde{\theta}=\frac{n^{-3/2}}{2\sigma\sqrt{2}}=\frac{n^{-3/2}}{\sqrt{10/3-\pi}}$. 

$\sum_{j}\langle\phi|U\rangle_{j,j+1,j+2}\langle U|\phi\rangle$
is nonzero only if it relates two walks that differ by a local move
at the positions $j,j+1,j+2$. In particular,
\begin{equation}
\sum_{j}\langle\phi|U\rangle_{j,j+1,j+2}\langle U|\phi\rangle=\frac{1}{2C_{n}}\sum_{j=1}^{2n-2}\sum_{\mathsf{s,t}\in D}e^{2\pi i\left(\tilde{A}_{s}-\tilde{A}_{t}\right)\tilde{\theta}}\langle\mathsf{s}|U\rangle_{j,j+1,j+2}\langle U|\mathsf{t}\rangle.\label{eq:upperbound_localmove1}
\end{equation}

The change in the area is either zero or one. There are three types
of nonzero contributions per $j,j+1,j+2$ in the foregoing equation
\[
\begin{array}{ccc}
\mathsf{s}_{j,j+1,j+2}=\mathsf{t}_{j,j+1,j+2} & : & e^{2\pi i\left(\tilde{A}_{s}-\tilde{A}_{t}\right)\tilde{\theta}}\langle\mathsf{s}|U\rangle_{j,j+1,j+2}\langle U|\mathsf{t}\rangle=1,\\
\mathsf{s}_{j,j+1,j+2}=udu\quad\mathsf{t}_{j,j+1,j+2}=uud & : & e^{2\pi i\left(\tilde{A}_{s}-\tilde{A}_{t}\right)\tilde{\theta}}\langle\mathsf{s}|U\rangle_{j,j+1,j+2}\langle U|\mathsf{t}\rangle=-e^{-4\pi i\tilde{\theta}},\\
\mathsf{s}_{j,j+1,j+2}=uud\quad\mathsf{t}_{j,j+1,j+2}=udu & : & e^{2\pi i\left(\tilde{A}_{s}-\tilde{A}_{t}\right)\tilde{\theta}}\langle\mathsf{s}|U\rangle_{j,j+1,j+2}\langle U|\mathsf{t}\rangle=-e^{4\pi i\tilde{\theta}}.
\end{array}
\]

The dependence is only on the difference of $\tilde{A}_{s}-\tilde{A}_{t}$
which is zero or two. Using the values of these three cases in Eq.
\ref{eq:upperbound_localmove1}, we find
\begin{equation}
\sum_{j=1}^{2n-2}\langle\phi|U\rangle_{j,j+1,j+2}\langle U|\phi\rangle=\frac{1}{C_{n}}\sum_{j=1}^{2n-2}a_{j}\left[1-\cos(4\pi\tilde{\theta})\right]\approx\frac{1}{C_{n}}\sum_{j=1}^{2n-2}8\pi^{2}\tilde{\theta}^{2}a_{j}\label{eq:UpperB_1}
\end{equation}
where $a_{j}$ is the number of strings that have $uud$ or $udu$
in their $j,j+2,j+2$ positions. An entirely a similarly calculation
gives
\begin{equation}
\sum_{j}\langle\phi|D\rangle_{j,j+1,j+2}\langle D|\phi\rangle=\frac{1}{C_{n}}\sum_{j=1}^{2n-2}b_{j}\left[1-\cos(4\pi\tilde{\theta})\right]\approx\frac{1}{C_{n}}\sum_{j=1}^{2n-2}8\pi^{2}\tilde{\theta}^{2}b_{j}\label{eq:UpperB_2}
\end{equation}
where $b_{j}$ is the number of strings that have $dud$ or $udd$
in their $j,j+2,j+2$ positions. Summing up Eq. \eqref{eq:UpperB_1}
and Eq. \eqref{eq:UpperB_2} and recalling that $\tilde{\theta}=\frac{n^{-3/2}}{\sqrt{10/3-\pi}}$, 
we have
\[
\langle\phi|H|\phi\rangle=\frac{8\pi^{2}}{C_{n}}\left(\frac{n^{-3}}{10/3-\pi}\right)\sum_{j=1}^{2n-2}a_{j}+b_{j}
\]
Since $(a_{j}+b_{j})/C_{n}=\mathcal{O}(1)$ , we have that $\langle\phi|H|\phi\rangle=\mathcal{O}(n^{-2})$.
If we take $s\ge1$, then similar arguments give
\[
\langle\phi|H|\phi\rangle\propto\frac{n^{-3}}{C_{n}s^{n}}\sum_{j=1}^{2n-2}\tilde{a}_{j}+\tilde{b}_{j}\sim O(n^{-2}).
\] \end{proof}
\subsubsection{The lower bound on the gap in balanced subspace is $O(n^{-15/2})$}
For now let us be restricted to the balanced subspace, where all the walks
in the superposition start at height zero and end at height zero and
never become negative. Moreover all up steps and down steps
have correctly ordered matching colors. 

In this subspace, $|{\cal D}^{s}\rangle$ is the unique ground state
of the frustration free Hamiltonian
\[
H=\sum_{j=1}^{2n-2}\mathbb{I}_{d^{j-1}}\otimes\left\{ \sum_{k_{1},k_{2}=1}^{s}|U^{k_{1},k_{2}}\rangle_{j,j+1,j+2}\langle U^{k_{1},k_{2}}|+|D^{k_{1},k_{2}}\rangle_{j,j+1,j+2}\langle D^{k_{1},k_{2}}|\right\} \otimes\mathbb{I}_{d^{2n-j-2}},
\]
where $\Pi_{boundary}$ and $\sum_{j=1}^{2n-1}\Pi_{j,j+1}^{cross}$
automatically vanish in the balanced subspace.

The Hamiltonian can be mapped to a classical Markov chain, denoted by the matrix $P$, with the entries
\begin{equation}
P(\mathsf{t,s})=\delta_{\mathsf{s,t}}-\frac{1}{2s(n-1)}\sqrt{\frac{\pi(\mathsf{s})}{\pi(\mathsf{t})}}\langle\mathsf{t}|H|\mathsf{s}\rangle,\label{eq:Pst_raw}
\end{equation}
where $\pi(\mathsf{s})\equiv\langle\mathsf{s}|{\cal D}\rangle^{2}=s^{-n}C_{n}^{-1}$
is the stationary distribution. Note that the gap of the Hamiltonian
is now related to the spectral gap of the Markov chain (because of
the minus sign). Denote the gap by $\Delta(H)$. Since they
are related by a similarity transformation, we have that
\begin{equation}
\Delta(H)=2s(n-1)(1-\lambda_{2}(P)).\label{eq:GapDefinition}
\end{equation}

Since $\pi(\mathsf{s})=\pi(\mathsf{t})$, the Markov chain takes the
simpler form 
\begin{eqnarray}
P(\mathsf{t,s}) & = & \delta_{\mathsf{s,t}}-\frac{1}{2s(n-1)}\langle\mathsf{t}|H|\mathsf{s}\rangle,\label{eq:P_st}
\end{eqnarray}
\begin{lem}
$P$ defined by Eq. \eqref{eq:P_st} is a reversible Markov chain
\end{lem}
\begin{proof}
1. $P$ is stochastic. The row sums are equal to one: 
\begin{eqnarray*}
\sum_{\mathsf{s}}P(\mathsf{t,s}) & = & \sum_{\mathsf{s}}\{\delta_{\mathsf{s,t}}-\frac{1}{2s(n-1)}\langle\mathsf{t}|H|\mathsf{s}\rangle\}=1-\sum_{\mathsf{s}}\frac{\pi(\mathsf{t})^{-1/2}}{2s(n-1)}\langle\mathsf{t}|H|\mathsf{s}\rangle\sqrt{\pi(\mathsf{s})}\\
 & = & 1-\frac{\pi(\mathsf{t})^{-1/2}}{2s(n-1)}\sum_{\mathsf{s}}\langle\mathsf{t}|H|\mathsf{s}\rangle\langle\mathsf{s}|{\cal D}\rangle=1-\frac{\pi(\mathsf{t})^{-1/2}}{2s(n-1)}\sum_{\mathsf{s}}\langle\mathsf{t}|H|{\cal D}\rangle=1
\end{eqnarray*}
where we used the completeness $\sum_{\mathsf{s}}|\mathsf{s}\rangle\langle\mathsf{s}|=1$
and the fact that $H|{\cal D}^{s}\rangle=0$ as $|{\cal D}^{s}\rangle$
is the zero energy ground state of $H$. Therefore, $P$ has row sums
equal to one.

2. $P$ has a unique stationary state:
\begin{eqnarray*}
\sum_{\mathsf{t}}\pi(\mathsf{t})P(\mathsf{t,s}) & = & \pi(\mathsf{s})-\sum_{\mathsf{t}}\frac{1}{2s(n-1)}\sqrt{\pi(\mathsf{s})\pi(\mathsf{t})}\langle\mathsf{t}|H|\mathsf{s}\rangle=\pi(\mathsf{s})-\sum_{\mathsf{t}}\frac{1}{2s(n-1)}\langle\mathsf{t}|{\cal D}^{s}\rangle\langle\mathsf{t}|H|\mathsf{s}\rangle\langle\mathsf{s}|{\cal D}^{s}\rangle\\
 & = & \pi(\mathsf{s})-\sum_{\mathsf{t}}\frac{1}{2s(n-1)}\frac{1}{\sqrt{C_{n}}s^{n/2}}\langle\mathsf{t}|H|\mathsf{s}\rangle\langle\mathsf{s}|{\cal D}^{s}\rangle=\pi(\mathsf{s})-\frac{1}{2s(n-1)}\langle\mathsf{{\cal D}}^{s}|H|\mathsf{s}\rangle\langle\mathsf{s}|{\cal D}^{s}\rangle=\pi(\mathsf{s})
\end{eqnarray*}

3. $P$ is reversible. Noting that $\pi(\mathsf{t})=\pi(\mathsf{s})$
and , $\langle\mathsf{t}|H|\mathsf{s}\rangle=\langle\mathsf{s}|H|\mathsf{t}\rangle,$
it is easy to check that $\pi(\mathsf{t})P(\mathsf{t,s})=\pi(\mathsf{s})P(\mathsf{s,t})$. 
\end{proof}
\begin{lem}
$P(\mathsf{s},\mathsf{s})\ge1/2$ and for $\mathsf{s}\ne\mathsf{t}$,
$P(\mathsf{t},\mathsf{s})$ is nonzero if only if $\mathsf{t}$ and
$\mathsf{s}$ differ by moving a peak by one position, in which case
its value is 
\[
P(\mathsf{t},\mathsf{s})=\frac{1}{4s(n-1)}.
\]
\end{lem}
\begin{proof}
From Eq. \eqref{eq:P_st} we have $P(\mathsf{s},\mathsf{s})=1-\frac{1}{2s(n-1)}\langle\mathsf{s}|H|\mathsf{s}\rangle\ge1/2$
since $\langle\mathsf{s}|H|\mathsf{s}\rangle\le(n-1)$. Now take $\mathsf{t}\ne\mathsf{s}$:
it is clear that if $\mathsf{s}$ differs from $\mathsf{t}$ in more
than three consecutive positions then $P(\mathsf{s},\mathsf{t})=0$.
Therefore $P(\mathsf{s},\mathsf{t})$ is nonzero if and only if $\mathsf{t}$
is obtained from $\mathsf{s}$ by moving a single peak by one position.
Now suppose $\mathsf{s}$ is obtained from $\mathsf{t}$ by moving
a single peak. Then for any local move that moves only a single peak
and leaves the rest of the positions equal we have $\langle\mathsf{s}|H|\mathsf{t}\rangle=-1/2$.
Suppose $\mathsf{s}$ and $\mathsf{t}$ differ at positions $j,j+1,j+2$,
then 
\begin{eqnarray}
P(\mathsf{s},\mathsf{t}) & = & -\frac{1}{2s(n-1)}\sum_{j=1}^{2n-2}\langle\mathsf{s}|\mathbb{I}_{d^{j-1}}\otimes\left\{ |U\rangle_{j,j+1,j+2}\langle U|+|D\rangle_{j,j+1,j+2}\langle D|\right\} \otimes\mathbb{I}_{d^{2n-j-2}}|\mathsf{t}\rangle\label{eq:Pst_nonzero}\\
 & = & -\frac{1}{2s(n-1)}\langle\mathsf{s}|\mathbb{I}_{d^{j-1}}\otimes\left\{ |U\rangle_{j,j+1,j+2}\langle U|+|D\rangle_{j,j+1,j+2}\langle D|\right\} \otimes\mathbb{I}_{d^{2n-j-2}}|\mathsf{t}\rangle=\frac{1}{4s(n-1)}.\nonumber 
\end{eqnarray}\end{proof}
Using Eqs. \eqref{eq:Gap_Markov_Final} and \eqref{eq:GapDefinition}
the lower bound on the gap in the balanced subspace is
\begin{equation}
\Delta(H)\ge O(s n^{-15/2}).\label{eq:GapLowerBound}
\end{equation}
It remains to lower bound the smallest eigenvalues restricted to the unbalanced subspace.
\subsubsection{The smallest energy in the unbalanced subspace is $O(n^{-c})$, $c\gg 1$}
Above we proved a polynomial gap to the first excited state by restricting
the Hamiltonian to a balanced subspace where no penalties result
from boundary terms (i.e., height imbalance) or any of the $\Pi_{j,j+1}^{cross}$'s
(i.e., mismatching of colors). We now prove a polynomially small lower
bound on the energies of the states in the unbalanced subspace.
\begin{lem}\label{Lem:unbalancedLowerbound}
The ground state energies (smallest eigenvalues) of the Hamiltonian restricted to the unbalanced subspace are at least $O(n^{-c})$ with $c\gg 1$.
\end{lem}
\begin{proof}
We separate the proof into two parts: 1. height-balanced subspace
with only mismatched colors and 2. height-imbalance subspace without
any mismatches. Simultaneous occurrence
of these penalties only increases the energy. 

\textbf{\textit{Color mismatching only:}} In this case the penalty
is only due to $\sum_{j}\Pi_{j,j+1}^{cross}$ as all the states have
zero initial and final heights and they automatically vanish at the
boundaries. The Hamiltonian restricted to this space is
\[
H=\sum_{j}\Pi_{j,j+1,j+2}+\sum_{j}\Pi_{j,j+1}^{cross}.
\]
 Assume we have the minimum (a single mismatch) resulting in the smallest
energy penalty. For example we can have $w_{0}u^{1}w_{1}u^{2}w_{3}d^{1}w_{4}d^{2}w_{5}$,
where $w_{0},\dots,w_{5}$ are balanced $s-$colored Dyck walks and
$u^{1}$ and $u^{2}$ are step ups of two different colors being balanced
by $d^{1}$ and $d^{2}$, respectively (hence a mismatch of colors).
Suppose $u^{2}$ is at the site $i$ and $d^{2}$ is at the site $k$.
Now take the amplitude of $\Pi_{i,i+1,i+2}$ and $\Pi_{i-1,i,i+1}$
to zero which would pin $u^{2}$ to remain at site $i$ and overall
lower the energy. Also let's take the amplitude of $\Pi_{k,k+1,k+2}$
and $\Pi_{k-1,k,k+1}$ to zero as to pin $d^{2}$ at site $k$. This
reduces the problem to a chain of length $k-i$ with a single imbalance
$d^{1}$. Therefore, a polynomial lower bound on the smallest energy
of height-imbalance (only) subspaces lower bounds the purely  mismatched
subspace. So we now turn to lower bounding the smallest energy of
the subspace with only a height-imbalanced. \\

\textbf{\textit{Height imbalance only:}} In this subspace $\Pi_{j,j+1}^{cross}$
automatically vanishes on any state. So the Hamiltonian restricted
to this subspace is
\[
H=\sum_{j}\Pi_{j,j+1,j+2}+\Pi_{boundary}.
\]

Our goal is to give $n^{-\mathcal{O}(1)}$ lower bound to the minimum
energy in the subspace where there are $p>0$ extra down and $q>0$
extra up steps. Any string on the imbalanced subspace is uniquely
written as
\[
\mathsf{s}=w_{0}d^{\centerdot}w_{1}d^{\centerdot}\dots d^{\centerdot}w_{p}w_{0}u^{\centerdot}v_{1}u^{\centerdot}v_{2}\dots u^{\centerdot}v_{q}
\]
where $w_{i}$ and $v_{i}$ are colored-Dyck walks and $d^{\centerdot}$
and $u^{\centerdot}$ denote the unmatched steps down and up of any
color, respectively. Since Dyck walks are of even length, it is clear
that $p+q$ must be even. 

Because of symmetry it is sufficient to consider $p>0$ only, where
there are only imbalance down steps. We can omit the boundary terms
$|u^{k}\rangle_{2n}\langle u^{k}|$, which only decreases the energy.
Our Hamiltonian becomes
\[
H=\sum_{j}\Pi_{j,j+1,j+2}+\sum_{k=1}^{s}|d^{k}\rangle_{1}\langle d^{k}|,
\]
where $\Pi_{j,j+1,j+2}$ is as before. Note that for a chain of length
$2n$ the number of imbalances needs to be even. Assume we have the
minimum violation due to a height imbalance at the boundary and let's
focus on a string with a single step up imbalance denoted by $w_{0}d^{k}w_{1}d^{k}w_{2}$,
where as before $w_{0}$, $w_{1}$ and $w_{2}$ are balanced $s-$colored
Dyck walks and $d^{k}$ is a step of color $k\in\{1,\dots,s\}$. This
allows us to drop $\sum_{k=1}^{s}|u^{k}\rangle_{2n}\langle u^{k}|$,
which only decreases the energy. Let us denote the first imbalanced
$d^{k}$ by $x$ and the rest of the imbalanced steps by $y$, then
the Hamiltonians is
\[
H^{x}=|x\rangle_{1}\langle x|+\sum_{j}\Pi_{j,j+1,j+2}+\Theta_{j,j+1,j+2}^{x}+\Theta_{j,j+1,j+2}^{y},
\]
where $\Theta_{j,j+1,j+2}^{x}$ is the projector that spans $\sum_{k=1}^{s}\frac{1}{\sqrt{2}}\left(|u^{k}d^{k}x\rangle-|xu^{k}d^{k}\rangle\right)$
and $\Theta_{j,j+1,j+2}^{y}$ projects onto $\sum_{k=1}^{s}\frac{1}{\sqrt{2}}\left(|u^{k}d^{k}y\rangle-|yu^{k}d^{k}\rangle\right)$.
We only have $|x\rangle\langle x|$ as a $y$ step can never pass
to the left of the $x$ down step. We think of $x$ and $y$ steps
as 'particles' that can move on the chain. In any configuration of
$y$ particles, the $x$ particle vanishes upon touching any $y$;
therefore $y$ particles serve as kind of domain walls for $x$.
Since our goal is to get a lower bound on the energy, we can only analyze
the interval between $1$ and the first $y$ particle which is located in
the position $m$ which is between the position of $x$ and $2n$.
So we can equivalently define the length of the chain to be $m$ and
analyze a single imbalance (i.e., $x$) on this chain. Since, we have
a single imbalance, $m$ necessarily is odd. 

The Hilbert space now is the span of
\[
|\mathsf{s}\rangle\otimes|x\rangle_{j}\otimes|\mathsf{t}\rangle,\quad\text{where }\text{\ensuremath{\mathsf{s}}}\in{\cal D}_{j-1}^{s}\text{ and }\mathsf{t}\in{\cal D}_{m-j}^{s}
\]
where $j$ is odd and ${\cal D}_{k}^{s}$ denotes the set of all $s-$colored
Dyck walks of length $k$. We can now use perturbation theory. Define
$H_{\epsilon}^{x}$ by 
\[
H_{\epsilon}^{x}=\sum_{j}\Pi_{j,j+1,j+2}+\epsilon\left\{ \sum_{j}\Theta_{j,j+1,j+2}^{x}+|x\rangle_{1}\langle x|\right\} ,
\]
where $0<\epsilon\le1$. Since the Hamiltonian is a sum of projectors
we have $\Delta(H^{x})\ge\Delta(H_{\epsilon}^{x})$. The spectral
gap of $H_{0}^{x}=\sum_{j}\Pi_{j,j+1,j+2}$ is $1/\text{poly}(m)$
as we now show. This is sufficient because in principle $m$ can be
as large as $2n-1$. 

First note that the position of $x-$particle at $j$ is an invariant
of $H_{0}^{x}$. Also any $\Pi_{j,j+1,j+2}$ vanishes at $j$ where
the $x-$particle is. So we can analyze $H_{0}^{x}$ separately on
the  disjoint intervals $[1,j-1]$ and $[j+1,m]$. The ground subspace
of $H_{0}^{x}$ is spanned by normalized states 
\begin{equation}
|\omega_{j}\rangle=|{\cal D}_{j-1}^{s}\rangle\otimes|x\rangle_{j}\otimes|{\cal D}_{m-j}^{s}\rangle.\label{eq:spanUnbalanced}
\end{equation}
The gap of $H_{0}^{x}$ can be computed separately on these intervals.
But above we showed that the gap of the Hamiltonian in the balanced
$s-$colored Dyck space is polynomially small in the length of the
chain. Therefore, we have that $\Delta(H_{0}^{x})\ge n^{-\mathcal{O}(1)}$. 

Suppose now $\epsilon>0$, the first order Hamiltonian acting on $|\omega_{j}\rangle$
states describes a ``hopping'' of the particle $x$ over $u^{k}d^{k}$
pairs on the chain of length $2n$ with a delta potential at $j=1$. 

Comment: The hopping is nonstandard in that the $x$ particle hops
two positions to the right (left)  and the $u^{k}d^{k}$ pair moves
one position to the left (right). 

The parameters of the hopping Hamiltonian can now be obtained. Let
us , as before, denote the number of the $s-$colored Dyck walks in
$|{\cal D}_{k}^{s}\rangle$ by $D_{k}^{s}=s^{k}C_{k}$, where $C_{k}$
is the $k^{th}$ Catalan number. The diagonal terms are:
\begin{eqnarray*}
\alpha_{j}^{2}\equiv\langle\omega_{j}|\Theta_{j,j+1,j+2}^{x}|\omega_{j}\rangle & = & \frac{s}{2}\frac{s^{m-j-2}C_{m-j-2}}{s^{m-j}C_{m-j}}=\frac{1}{2s}\frac{C_{m-j-2}}{C_{m-j}}\\
\langle\omega_{j+1}|\Theta_{j,j+1,j+2}^{x}|\omega_{j+1}\rangle & = & 0\\
\beta_{j}^{2}\equiv\langle\omega_{j+2}|\Theta_{j,j+1,j+2}^{x}|\omega_{j+2}\rangle & = & \frac{1}{2s}\frac{C_{j-1}}{C_{j+1}},
\end{eqnarray*}
and off-diagonal terms are
\begin{eqnarray*}
\langle\omega_{j}|\Theta_{j,j+1,j+2}^{x}|\omega_{j+1}\rangle & = & 0\\
-\alpha_{j}\beta j\equiv\langle\omega_{j}|\Theta_{j,j+1,j+2}^{x}|\omega_{j+2}\rangle & = & -\frac{s}{2}\frac{\sqrt{D_{j-1}^{s}D_{m-j-2}^{s}}}{\sqrt{D_{j+1}^{s}D_{m-j}^{s}}}=-\frac{1}{2s}\frac{\sqrt{C_{j-1}C_{m-j-2}}}{\sqrt{C_{j+1}C_{m-j}}}.
\end{eqnarray*}
We arrive at the effective next nearest neighbors hopping Hamiltonian
acting on $\mathbb{C}^{m}$:
\begin{equation}
H_{eff}=|1\rangle\langle1|+\sum_{j=1}^{m-2}\Gamma_{j,j+1,j+2}\label{Eq:Heff}
\end{equation}
where $\Gamma_{j,j+1,j+2}$ is a rank-$1$ operator 
\begin{eqnarray*}
\Gamma_{j,j+1,j+2} & = & \alpha_{j}^{2}|j\rangle\langle j|+\beta_{j}^{2}|j+2\rangle\langle j+2|\\
 &  & -\alpha_{j}\beta_{j}\left\{ |j\rangle\langle j+2|+|j+2\rangle\langle j|\right\} .
\end{eqnarray*}

For example the matrix representation of $\sum_{j=1}^{m-2}\Gamma_{j,j+1,j+2}$
for $m=5$ is
\[
\left[\begin{array}{ccccc}
\alpha_{1}^{2} & 0 & -\alpha_{1}\beta_{1} & 0 & 0\\
0 & \alpha_{2}^{2} & 0 & -\alpha_{2}\beta_{2} & 0\\
-\alpha_{1}\beta_{1} & 0 & \alpha_{3}^{2}+\beta_{1}^{2} &  & -\alpha_{3}\beta_{3}\\
0 & -\alpha_{2}\beta_{2} & 0 & \beta_{2}^{2} & 0\\
0 & 0 & -\alpha_{3}\beta_{3} & 0 & \beta_{3}^{2}
\end{array}\right]=\frac{1}{2s}\left[\begin{array}{ccccc}
\frac{C_{2}}{C_{4}} & 0 & -\sqrt{\frac{C_{0}}{C_{4}}} & 0 & 0\\
0 & \frac{C_{1}}{C_{3}} & 0 & -\frac{C_{1}}{C_{3}} & 0\\
-\sqrt{\frac{C_{0}}{C_{4}}} & 0 & 2\frac{C_{0}}{C_{2}} &  & -\sqrt{\frac{C_{0}}{C_{4}}}\\
0 & -\frac{C_{1}}{C_{3}} & 0 & \frac{C_{1}}{C_{3}} & 0\\
0 & 0 & -\sqrt{\frac{C_{0}}{C_{4}}} & 0 & \frac{C_{2}}{C_{4}}
\end{array}\right]
\]

Ignoring the repulsive delta potential at $j=1$, i.e., $|1\rangle\langle1|$,
we have $H_{move}=\sum_{j=1}^{m-2}\Gamma_{j,j+1,j+2}$, which is frustration
free. It has the unique zero energy ground state that is:
\begin{equation}
|g\rangle\propto s^{\frac{m-1}{2}}\sum_{j=1}^{m}\sqrt{C_{j-1}C_{m-j}}\text{ }|j\rangle.\label{eq:g}
\end{equation}

The gap of $H_{\epsilon}^{x}$ can be related to the gap of $H_{eff}$
via the projection lemma \cite{KempeKitaevRegev}. This lemma says
that 
\[
\lambda_{1}(H_{\epsilon}^{x})\ge\epsilon\lambda_{1}(H_{eff})-\frac{\mathcal{O}(\epsilon^{2})\left\Vert K\right\Vert }{\lambda_{2}(H_{0}^{x})-2\epsilon\left\Vert K\right\Vert },
\]
where $K=|x\rangle_{1}\langle x|+\sum_{j=1}^{m-2}\Theta_{j,j+1,j+2}^{x}$
is the perturbation operator. Since $\lambda_{2}(H_{0}^{x})\ge m^{-O(1)}$,
we can choose $\epsilon$ polynomially small in $1/m$ such that $2\epsilon\left\Vert K\right\Vert $
is small compared to $\lambda_{2}(H_{0}^{x})$. With this choice of
$\epsilon$ one gets
\[
\lambda_{1}(H_{\epsilon}^{x})\ge\epsilon\lambda_{1}(H_{eff})-\mathcal{O}(\epsilon^{2})m^{\mathcal{O}(1)}.
\]
The problem reduces to showing that $\lambda_{1}(H_{eff})\ge m^{-\mathcal{O}(1)},$
where $H_{eff}$ is the single particle hopping Hamiltonian defined
by Eq. \eqref{Eq:Heff}.

We now use the projection lemma to bound the gap of $H_{eff}$ by
first bounding the gap of $H_{move}$ and then treating $|1\rangle\langle1|$
as a perturbation to $H_{move}$. As before we map $H_{move}$ onto
a classical stochastic matrix that describes a random walk on $[1,m]$
with stationary distribution $\pi(j)\equiv\langle j|g\rangle^{2}$:
\[
P(j,k)=\delta_{j,k}-\sqrt{\frac{\pi(k)}{\pi(j)}}\langle j|H_{move}|k\rangle.
\]
Similar to before, $\sqrt{\pi(k)}|k\rangle$ is in the kernel of $H_{move}$
and we have $\sum_{k}P_{j,k}=1$ and $\sum_{j}\pi(j)P(j,k)=\pi(k)$.
The entries of $P(j,k)$ can be bounded. The off diagonal elements
are simply
\begin{eqnarray*}
P(j,j+2) & = & -\frac{\langle j+2|g\rangle}{\langle j|g\rangle}\langle j|H_{move}|j+2\rangle=\frac{\sqrt{C_{j+1}C_{m-j-2}}}{\sqrt{C_{j-1}C_{m-j}}}\alpha_{j}\beta_{j}=\frac{1}{2s}\frac{C_{m-j-2}}{C_{m-j}}\\
P(j+2,j) & = & -\frac{\langle j|g\rangle}{\langle j+2|g\rangle}\langle j+2|H_{move}|j\rangle=\frac{\sqrt{C_{j-1}C_{m-j}}}{\sqrt{C_{j+1}C_{m-j-2}}}\alpha_{j}\beta_{j}=\frac{1}{2s}\frac{C_{j-1}}{C_{j+1}}
\end{eqnarray*}

The Catalan numbers have the property that asymptotically they grow
as $C_{k}\sim\frac{4^{k}}{k^{3/2}\sqrt{\pi}}$ and for any $k\ge1$
one has $1/16\le C_{k}/C_{k+2}\le1$, which implies that 
\[
\frac{1}{32s}\le P(j,j\pm2)\le\frac{1}{2s}\qquad\forall j.
\]
This implies that the diagonal entries $P(j,j)\ge0$ and therefore
$P(j,k)$ is indeed a transition matrix from $j$ to $k$. The steady
state is nearly uniform because
\[
\frac{\pi(k)}{\pi(j)}\equiv\frac{\langle k|g\rangle}{\langle j|g\rangle}=\frac{C_{k-1}C_{m-k}}{C_{j-1}C_{m-j}}
\]
is upper bounded by $m^{\mathcal{O}(1)}$ and lower bounded by $m^{-\mathcal{O}(1)}$.
We can now bound the spectral gap of $P$ using the canonical path
theorem of Jerrum and Sinclair \cite{Sinclair1992}. This technique
gives $1-\lambda_{2}(P)\ge1/\rho L$ and $\rho$ is the maximum edge
load defined by \cite{Sinclair1992}
\begin{equation}
\rho=\max_{(a,b)\in E}\frac{1}{\pi(a)P(a,b)}\sum_{\mathsf{s},\mathsf{t}\text{:}(a,b)\in\gamma_{\mathsf{s},\mathsf{t}}}\pi(\mathsf{s})\pi(\mathsf{t}),\label{eq:Rho_CanonicalPath-1}
\end{equation}
where $a,b$ are arbitrary vertices and $\gamma_{\mathsf{s},\mathsf{t}}$
is a canonical path and we take $L\equiv\max_{\mathsf{s},\mathsf{t}}|\gamma_{\mathsf{s},\mathsf{t}}|$.
This formula resembles Eq. \eqref{eq:A} in the comparison theorem;
however, it is different because it defines the canonical paths on
the same Markov chain.

The Canonical path $\gamma(\mathsf{s},\mathsf{t})$ simply moves $x$
from $\mathsf{s}$ to $\mathsf{t}$. With the lower bound of $m^{-\mathcal{O}(1)}$,
in the denominator of Eq. \eqref{eq:Rho_CanonicalPath-1}, we arrive
at $1-\lambda_{2}(P)\ge m^{-\mathcal{O}(1)}$. It proves that $\lambda_{2}(H_{Move})\ge m^{-\mathcal{O}(1)}$. 

To finish the proof we apply the projection lemma to  Eq. \eqref{Eq:Heff},
treating $|1\rangle\langle1|$ as a perturbation. Since 
\[
\langle1|g\rangle^{2}=\pi(1)=\frac{s^{m-1}C_{0}C_{m}}{s^{m-1}\sum_{j=1}^{m}C_{j-1}C_{m-j}},
\]
using the recursion property of Catalan numbers, we conclude that
$\langle1|g\rangle^{2}=C_{m}/C_{m+1}$, which is a constant. This
proves that $\lambda_{1}(H_{eff})\ge m^{-\mathcal{O}(1)}$. Moreover, recall from the discussion above, that this lower
bound is sufficient for lower bounding the ground state when there
is mismatches in the colors of the paths as well. This completes the proof.\end{proof}
Lemmas \eqref{Lem:UpperBound} and \eqref{Lem:unbalancedLowerbound}, and Eq. \eqref{eq:GapLowerBound} prove the Theorem in  \eqref{MainTheorem}.  

\section{Conclusions}
We have proved that the energy gap of Fredkin spin chain is $\Delta(H)=O(n^{-c})$ with $c\ge 2$. This in particular implies that the continuum limit of this model will not be a conformal field theory either.

We now have quantum spin chain models that violate the area law exponentially more than previously thought possible. This  started with \cite{movassagh2016supercritical}, which is a local and an integer spin (bosonic) chain model.  Here we proved that the corresponding fermionic model introduced  in  \cite{salberger2016fredkin,dell2016violation} also has a gap that closes slowly (polynomially small with the system's size) as in \cite{movassagh2016supercritical}. Violating the area law is not necessarily hard; however, proposing models that are local, translationally invariant with a unique ground state that violate the area law is more difficult.  An outstanding challenge is to find an exactly solvable model that violates the area law maximally (i.e., fact or of $n$) {\it and} has a polynomially small gap. So far the maximum violation has been achieved with an exponentially small (in system size squared) gap \cite{zhang2017novel,levine2016gap}.
\section{Acknowledgements}
I thank Peter W. Shor for fruitful discussions. I also thank Scott
Sheffield, Baruch Schieber, Lionel Levine, Sergey Bravyi and Richard
Stanley. I am grateful for the support and freedom provided by the
Herman Goldstine fellowship in mathematical sciences at IBM TJ Watson
Research Center. I acknowledge the support of the Simons Foundation
and the American Mathematical Society for the AMS-Simons travel grant.
\bibliographystyle{plain}
\bibliography{mybib}
\end{document}